\documentclass[prl,twocolumn,preprintnumbers,superscriptaddress,amsmath,amssymb]{revtex4-1}
\usepackage{graphicx}
\usepackage{subfigure}
\usepackage{mathrsfs}
\usepackage{amsfonts}
\usepackage{times}
\usepackage{amsmath}
\usepackage{amsthm}
\usepackage{leftidx}
\usepackage{tikz}
\usepackage{tikz-network}
\usepackage{color}
\usepackage[colorlinks,linkcolor=blue,citecolor=blue]{hyperref}

\newcommand{\Tr}{\operatorname{Tr}}

\newtheorem{theorem}{Theorem}

\usepackage{bbold}
\usepackage{braket}
\usepackage{mathtools}

\begin{document}
\title{Threefold Way for Typical Entanglement}
\author{Haruki Yagi}
\affiliation{Department of Applied Physics, University of Tokyo, 7-3-1 Hongo, Bunkyo-ku, Tokyo 113-0033, Japan}
\author{Ken Mochizuki}
\affiliation{Department of Applied Physics, University of Tokyo, 7-3-1 Hongo, Bunkyo-ku, Tokyo 113-0033, Japan}
\affiliation{Nonequilibrium Quantum Statistical Mechanics RIKEN Hakubi Research Team,
RIKEN Cluster for Pioneering Research (CPR), 2-1 Hirosawa, Wako 351-0198, Saitama, Japan}
\author{Zongping Gong}
\affiliation{Department of Applied Physics, University of Tokyo, 7-3-1 Hongo, Bunkyo-ku, Tokyo 113-0033, Japan}
\date{\today}

\begin{abstract}
A typical quantum state with no symmetry can be realized by letting a random unitary act on a fixed state, and the subsystem entanglement spectrum follows the Laguerre unitary ensemble (LUE). For integer-spin time reversal symmetry, we have an analogous scenario where we prepare a time-reversal symmetric state and let random orthogonal matrices act on it, leading to the Laguerre orthogonal ensemble (LOE). However, for half-integer-spin time reversal symmetry, a straightforward analogue leading to the Laguerre symplectic ensemble (LSE) is no longer valid due to that time reversal symmetric state is forbidden by the Kramers' theorem. We devise a system in which the global time reversal operator is \emph{fractionalized} on the subsystems, and show that LSE arises in the system. Extending this idea, we incorporate general symmetry fractionalization into the system, and show that the statistics of the entanglement spectrum is decomposed into a direct sum of LOE, LUE, and/or LSE. Here, various degeneracies in the entanglement spectrum may appear, depending on the non-Abelian nature of the symmetry group and the cohomology class of the non-trivial projective representation on the subsystem. Our work establishes the entanglement counterpart of the Dyson's threefold way for Hamiltonians with symmetries.
\end{abstract}
\maketitle

\emph{Introduction.---}
Entanglement is a genuine quantum property and has played an increasingly important role in physics \cite{Haroche2001,horodecki2008entanglement,Horodecki2009,Eisert2010colloquium,nielsen_chuang_2010,Pan2012,Harlow2016,Abanin2019,Cirac2021}.  
A renowned example in high-energy physics is the black hole information paradox, into which entanglement brought about crucial insights \cite{Mathur2009,Almheiri2021,Raju2022}. 
Historically, Page made a breakthrough by introducing the celebrated Page curve in 1990s, which captures the statistical property of information loss \cite{hawking1974black,hawking1976breakdown,bekenstein1981universal,page1993information,page1993average}.
Meanwhile, random matrix theory (RMT) was developed as a powerful tool for analysis \cite{page1993average,foong1994proof,sanchez1995simple,sen1996average,bianchi2022volume}. The black hole information paradox itself continued to be actively studied and debated ever after \cite{Hayden2007,Maldacena2013,Penington2020}.
Moreover, entanglement evaluation and the Page curve have impacted many other fields such as quantum thermalization and measurement-induced phase transitions \cite{Popescu2006,nahum2017quantum,bao2020theory,skinner2019measurement}.

The original Page curve deals with entanglement in completely random quantum states, as is consistent with the maximal chaotic nature of black holes \cite{Sekino2008,Maldacena2016,Cotler2017}. However, most physical systems, ranging from microscopic molecules to macroscopic materials, exhibit some symmetries \cite{Lax74}. 
An exceptional representative is the \emph{time reversal symmetry} (TRS), which is of particular privilege in RMT because it makes a qualitative difference to the matrix ensemble through specifying whether the matrix elements are real, complex or quaternionic \cite{forrester2010log}.
For Hamiltonians described by Hermitian matrices, the corresponding Gaussian orthogonal, unitary, and symplectic ensembles are well-known as the Dyson's threefold way \cite{dyson1962threefold}. For reduced density matrices described by Hermite semidefinite matrices, there are also three distinct matrix ensembles: 
the Laguerre orthogonal, unitary, and symplectic ensemble (LOE, LUE, and LSE).  
A typical reduced density matrix without symmetry, as is the case of the original Page curve, follows the LUE \cite{forrester2010log}. 
In the presence of integer-spin TRS, one naturally starts from random time-reversal symmetric states and can easily check a typical reduced density matrix follows the LOE. Remarkably, this picture breaks down for half-integer-spin TRS since \emph{no} eigenstate of TRS exists due to the Kramers' theorem \cite{Bernevig2013}. Accordingly, there appears to be no straightforwardly analogous physical interpretation of the remaining LSE. 

In this Letter, we unveil the entanglement interpretation of the LSE by exploiting the notion of \emph{symmetry fractionalization}, which is well-known in the studies of quantum anomaly and topological phases \cite{Chen2011,Pollmann2012,Delmastro2023}. Regarding the TRS, the idea is simply that an integer spin can be fractionalized into two half-integer spins for each subsystems, \`{a} la the edges of the Haldane phase \cite{Haldane1983}. 
Furthermore, this idea can be extended to more general symmetries and their anomalous (projective) realizations. We show that for any symmetry described by a finite group, the ensemble of the corresponding symmetric density matrices can always be decomposed into a direct sum of LOE, LUE, or/and LSE. Our work thus establishes the entanglement counterpart of the Dyson's threefold way \cite{dyson1962threefold}.

\emph{Entanglement spectrum and RMT.---}
In a bipartite system $a\cup b$, 
the \emph{entanglement spectrum} of a pure state $\ket{\Psi}$ is defined as the eigenvalues of the reduced density matrix \cite{Li2008,Fidkowski2010,pollmann2010entanglement}. Without loss of generality, we assume the Hilbert space dimension of $a$ is not larger than the dimension of $b$. Expanding the state as $\ket{\Psi}=\sum_{a,b}w_{ab}\ket{a}\ket{b}$, one can obtain the reduced density matrix on $a$ as $\rho_a=\Tr_b|\Psi\rangle\langle\Psi|=WW^\dagger$ with $W=\sum_{a,b}w_{ab}|a\rangle\langle b|$ \footnote{Here we abuse the same letter to refer to a subsystem and a quantum state living in it.}. When we are interested in the typical behavior of entanglement, we sample $\ket{\Psi}$ with no symmetry from the Haar measure on the total Hilbert space. 
Alternatively, we can sample $w_{ab}$ from the complex standard normal distribution $\mathcal{CN}(0,1)$ (and normalize to $\braket{\Psi|\Psi}=1$ later) \cite{zyczkowski2001induced,nechita2007asymptotics}, implying $\rho$ follows the LUE (up to normalization).  
Likewise, $\rho$ follows the LOE (LSE) if the random matrix elements $w_{ab}$ are replaced by random real numbers (quaternions).

Another equivalent way to explore typical entanglement is sampling $U$ from circular unitary ensemble (CUE) and adopt $U\ket{0}$ ($|0\rangle$: arbitrary fixed reference state) as a random state. This is exactly the picture underlying the original Page curve \cite{page1993information,page1993average}. The real-number/quaternion counterpart of CUE is known as the circular orthogonal ensemble (COE)/circular symplectic ensemble (CSE) \cite{forrester2010log}. 
It consists of random unitaries preserving the integer/half-integer spin TRS, which is represented by an anti-unitary operator $\mathcal{T}$ squaring to identity/minus identity \cite{Wigner1960}. If $\mathcal{T}^2=\mathbb{1}$, one can further choose $|0\rangle$ to be time-reversal symmetric, i.e., $|0\rangle=\mathcal{T}|0\rangle$, and obtain the LOE. If $\mathcal{T}^2=-\mathbb{1}$, however, the Kramers' theorem states that $\mathcal{T}|0\rangle$ is always orthogonal to $|0\rangle$, implying $|0\rangle$ can never be time-reversal symmetric \cite{Bernevig2013}. If we drop the symmetry constraint on $|0\rangle$, it is recently shown that $U|0\rangle$ generates the same state ensemble for CSE and CUE \cite{west2024randomensemblessymplecticunitary}. A direct analogue thus breaks down for the LSE. 

\emph{Symmetry fractionalization.---}
Noteworthily, $\rho_a$ following the LSE only implies $[\mathcal{T}_a,\rho_a]=0$ for some $\mathcal{T}_a^2=-\mathbb{1}_a$, where $\mathbb{1}_a$ is the identity operator on the subsystem $a$. If $\mathcal{T}^2_b=-\mathbb{1}_b$, the entire TRS $\mathcal{T}=\mathcal{T}_a\otimes\mathcal{T}_b$ could be involutory, i.e., $\mathcal{T}^2=(-\mathbb{1}_a)\otimes(-\mathbb{1}_b)=\mathbb{1}$, allowing the existence of globally time-reversal symmetric (pure) states. 
The underlying physical intuition is an integer spin can be decomposed into two half-integer spins. In contrast, the LOE corresponds to a more natural division of subsystems with $\mathcal{T}_a^2=\mathbb{1}_a$ and $\mathcal{T}_b^2=\mathbb{1}_b$, corresponding to a decomposition of an integer spin into two integer spins.

Let us consider the decomposition 
of TRS in further detail. By properly choosing the basis, we have $\mathcal{T}=\mathbb{1}_a\otimes\mathbb{1}_b K$ ($K$: complex conjugation) for the normal decomposition corresponding to the LOE, and
$\mathcal{T}=(\mathbb{1}'_a \otimes i\sigma_y)\otimes(i\sigma_y\otimes\mathbb{1}'_b)K$ ($\mathbb{1}_{a}=\mathbb{1}'_{a}\otimes\mathbb{1}_2$ and $\mathbb{1}_{b}=\mathbb{1}_2\otimes\mathbb{1}'_{b}$) for the spin fractionalization decomposition \cite{Wigner1960}, which hopefully gives the LSE. Note that these two decompositions can be switched into each other via a local unitary conjugation by \cite{cirac2017matrix} 
\begin{equation}
    \Upsilon=\mathbb{1}_a'\otimes\frac{1-i}{2}\left[\mathbb{1}_4-i\left(\sigma_y\otimes \sigma_y\right)\right]\otimes\mathbb{1}_b',
    \label{Ups}
\end{equation}
i.e., $\Upsilon \left[\mathbb{1}_a\otimes\mathbb{1}_b K\right] \Upsilon^\dag=\left(\mathbb{1}_a'\otimes i\sigma_y\right)\otimes\left(i\sigma_y\otimes \mathbb{1}_b'\right) K$. 
This fractionalization implies $[\mathbb{1}_L\otimes i\sigma_y K, \rho'_a]=0$, where $\rho'_a=\Tr_b\Upsilon|\Psi\rangle\langle\Psi|\Upsilon^\dag$.
One can check that $\Upsilon$ indeed turns a LOE into a LSE: 
We start from a random integer-spin TRS state $\ket{\Psi}$ 
represented as a vector with random real elements, therefore $\rho_a=WW^\dag$ produces the LOE. By applying $\Upsilon$ on $\ket{\Psi}$, we find four independent random real elements $\{x_j\}^4_{j=1}$ in a $2\times 2$ block $\begin{pmatrix}x_1&x_2\\x_3&x_4\end{pmatrix}$
of $W$ is rearranged into two independent complex elements $q_1=\frac{1-i}{2}(x_1+ix_2), q_2=\frac{1-i}{2}(x_3-ix_4)$ and their complex conjugates with proper signs.  
Namely, the corresponding $2\times 2$ block becomes $\begin{pmatrix}q_1&q_2\\-\overline{q}_2&\overline{q}_1\end{pmatrix}$, 
which is the spin representation of quaternion with independent coefficients \cite{SM}. This implies $\rho'_a=\Tr_b\Upsilon|\Psi\rangle\langle\Psi|\Upsilon^\dag$ produces the LSE.

This scenario can be generalized to symmetries other than TRS. 
Suppose the system exhibits a global symmetry described by group $G$ and represented by $\mathbf{D}(g)$, which may contain both unitary and anti-unitary operators. If the global representation admits a tensor product decomposition on subsystems $a$ and $b$, i.e., $\mathbf{D}(g)=\mathcal{D}_a(g)\otimes\mathcal{D}_b(g)$, then $\mathcal{D}_{a,b}(g)$ will generally be a \emph{projective} representation. Technically, if $\mathbf{D}(g)$ 
is a linear (anti-linear) representation, there exists a 2-cocycle $\omega$ such that $\mathcal{D}_a(g)\mathcal{D}_a(g')=\omega(g,g')\mathcal{D}_a(gg')$ and $\mathcal{D}_b(g)
\mathcal{D}_b(g')=\overline{\omega(g,g')}\mathcal{D}_b(gg')$ $\forall g,g'\in G$ ($\omega(g,g')\in U(1)$), where the overline denotes the complex conjugation. This is generally referred to as \emph{symmetry fractionalization}, which is well-known in the context of symmetry-protected topological (SPT) phases \cite{chen2013symmetry,zeng2019quantum}. 
Canonical examples include the Haldane spin chain \cite{Haldane1983} and the Affleck-Kennedy-Lieb-Tasaki (AKLT) model \cite{affleck1987rigorous,tasaki2020physics} with spin-1/2 quasiparticles at both edges under the open boundary condition even though the bulk elements are of spin-1. 
The projective realization of symmetry also gives the simplest example of 't Hooft anomaly (in ($0+1$) dimension) \cite{Delmastro2023}, which, just like the Kramers' theorem, forbids any pure state to be symmetric. In the following, we show how to incorporate the general scenario of symmetry fractionalization into RMT, and what ensembles the reduced density matrices obey.

\emph{General setup.---}
We consider general unitary symmetries described by a finite group $G_0$, possibly combined with the TRS $\mathbb{Z}_2^{\mathcal{T}}$. In the (former) latter case, the entire symmetry group reads ($G=G_0$) $G=G_0\rtimes \mathbb{Z}_2^\mathcal{T}$ with $\mathbb{Z}_2^\mathcal{T}=\{e,t\}$ and $\tilde g=tgt\in G_0$ $\forall g\in G_0$. As widely used to construct SPT phases and topological quantum field theories \cite{chen2013symmetry, zeng2019quantum, dijkgraaf1990topological}, the system is a lattice consisting of $N$ $G_0$-spins \cite{Szlachanyi1993}. Each local Hilbert space is spanned by an orthonormal basis $\{|g\rangle\}_{g\in G_0}$, which are transformed via $D(g)|h\rangle=|gh\rangle$ by an on-site (regular) unitary representation $D(g)$. 

In the presence of TRS, we further have $D(t)|g\rangle=|\tilde g\rangle$ with $D(t)$ being an on-site anti-unitary representation. We are interested in the typical entanglement spectrum of a $G$-symmetric many-body state $|\Psi\rangle=D(g)^{\otimes N}|\Psi\rangle$ $\forall g\in G$. 

Given a real-space bipartition $a\cup b$, we can always rearrange the sites into a one-dimensional configuration such that all the $G_0$-spins in $a$ ($b$) are on the left (right) side of a cut (see Fig.~\ref{fig:setup} (a) and (b)). Such a reconfiguration does not rely on the detail (e.g., dimension and topology) of the lattice, and has even become accessible in experiments \cite{Bluvstein2024}. Denoting the two $G_0$-spins nearest to the cut as $l$, $r$, and the remaining parts as $L=a\backslash l$ and $R=b\backslash r$, we can perform a unitary transformation that turns $D(g)^{\otimes N}$ into $\mathbb{1}_L\otimes D(g)\otimes D(g)\otimes \mathbb{1}_R$ while preserves the entanglement spectrum (see Fig.~\ref{fig:setup} (b)) \cite{SM}. Hereafter, we work in this rotating frame, which will make the symmetry fractionalization picture clearer while greatly simplify the analysis. Also, it allows the (Hilbert space) dimension of $L$ or $R$, which we will denote as $d_L$ or $d_R$ respectively, to be an arbitrary integer rather than a power of $|G_0|$.

\begin{figure}
\begin{center}
       \includegraphics[width=8.5cm, clip]{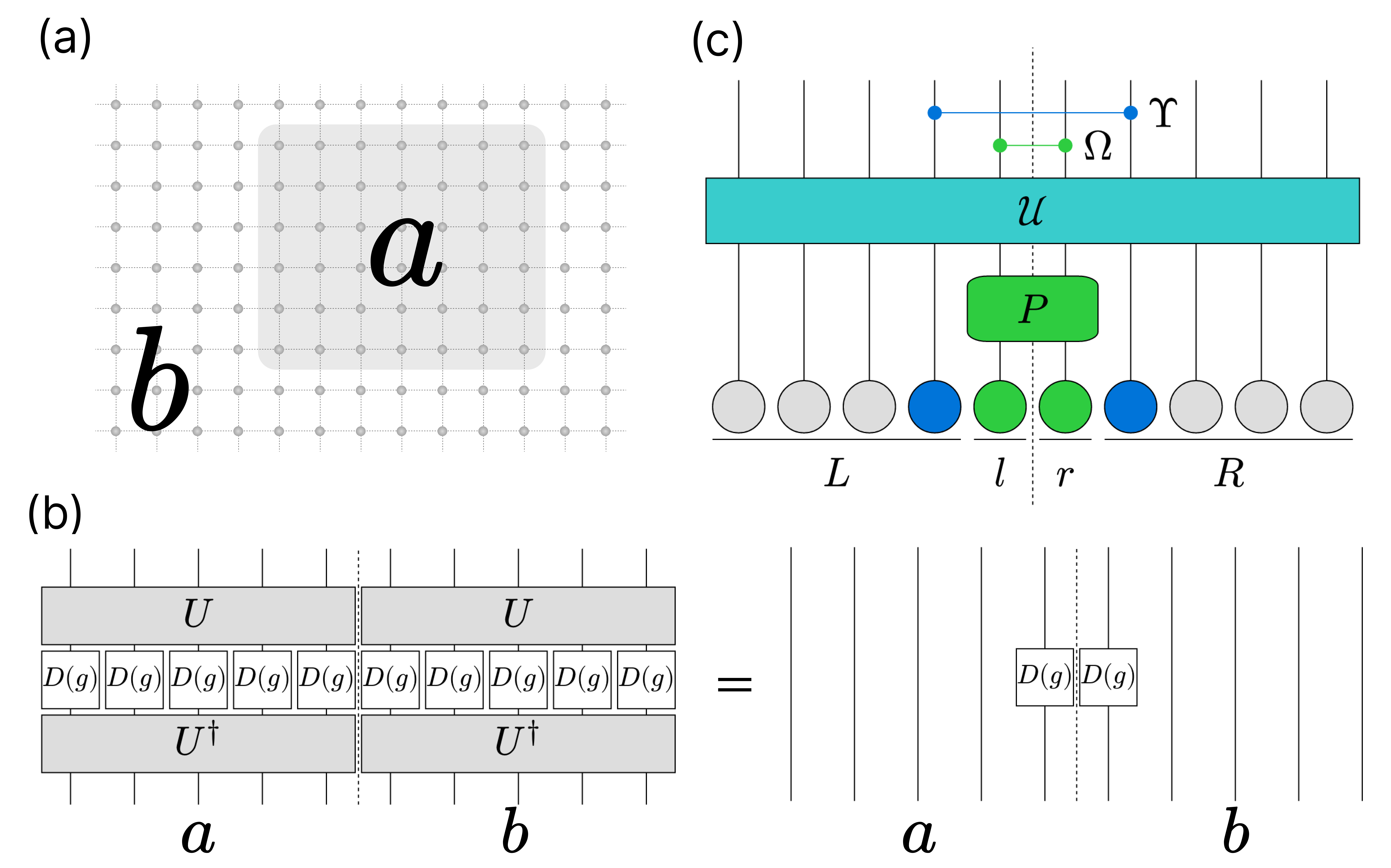}
       \end{center}
   \caption{(a) Bipartitite system $a\cup b$ in an arbitrary (e.g., two) spatial dimension. The system is supposed to have an onsite symmetry $G$.
   (b) Local unitary conjugation of the onsite symmetry $D(g)^{\otimes N}$ gives $\mathbb{1}_L\otimes D(g)\otimes D(g)\otimes \mathbb{1}_R$, shown in the rearranged one-dimensional configuration. (c) Configuration of the system consisting of $L,l,r,R$. 
   Here $P$ acting on $l\cup r$ is a projection from the $|G_0|^2$-dimensional space onto the $|G_0|$-dimensional $G_0$-symmetric space. Within this subspace, $\mathcal{U}$ is a unitary operator sampled from the Haar measure on some compact Lie groups. If $G=G_0$ ($G=G_0\rtimes\mathbb{Z}_2^\mathcal{T}$), that Lie group is the unitary group (isomorphic to the orthogonal group). Application of $\Omega$ (or/and $\Upsilon$) realizes the symmetry fractionalization of $G_0$ ($G_0\rtimes\mathbb{Z}_2^\mathcal{T}$).}
   \label{fig:setup}
\end{figure}

Let us identify the explicit form of $w_{ab}$ constrained by symmetry. Note that the projector onto the $G_0$-symmetric subspace of $l\cup r$ is given by $P=|G_0|^{-1}\sum_{g\in G_0}D(g)\otimes D(g)$, whose trace is $|G_0|$, which gives the dimension of the subspace. In fact, one can choose the orthonormal basis to be 
\begin{equation}
    \ket{\psi_g} = \frac{1}{\sqrt{|G_0|}} \sum_{h \in G_0} \ket{h g}\ket{h},\quad \forall g \in G_0. 
\end{equation}
These $G_0$-symmetric states are not necessarily invariant under time reversal (unless $G=G_0\times \mathbb{Z}^{\mathcal{T}}_2$):
$D(t) \otimes D(t)
\ket{\psi_g}= \ket{\psi_{\tilde{g}}}$.
In terms of $|\psi_g\rangle$, we can write down a random $G$-symmetric state of the entire system as
\begin{equation}
\begin{split}
    \ket{\Psi} &= \sum_{L,g,R} c_{L,g,R} \ket{L}\ket{\psi_g}\ket{R}\\
    &= \frac{1}{\sqrt{|G_0|}} \sum_{L,g_l,g_r,R} c_{L,g_r^{-1}g_l, R}\ket{L}\ket{g_l}\ket{g_r}\ket{R},
\end{split}
\label{Psi}
\end{equation}
where $c_{L,g,R}$'s are sampled independently from $\mathcal{CN}(0,1)$. In the presence of TRS, we further require $c_{L,\tilde{g},R}=\overline{c_{L,g,R}}$. In particular, this implies $c_{L,g,R}$ follows the real normal distribution $\mathcal{N}(0,1)$ if $\tilde g=g$ \footnote{This setup does not guarantee $\braket{\Psi|\Psi}=1$. After all random numbers $c_{L,g,R}$ have been sampled, a final normalization is performed to obtain a state with a norm of 1. Note that the discussion that follows describes the unnormalized states. Yet the distributions of entanglement spectra for normalized states are explicitly figured out in the Supplemental Material.}. So far, we have $w_{ab}\propto c_{L,g_r^{-1}g_l,R}$ ($a=L g_l$, $b=g_r R$) and $[\mathbb{1}_L\otimes D(g),\rho_a]=0$ $\forall g\in G$.

We move on to incorporate symmetry fractionalization. To impose a fractionalized (projective) symmetry to the subsystem, we only have to apply the following local unitary gate $\Omega$ (supported on $l\cup r$) to $\ket{\Psi}$ (\ref{Psi}) \cite{classification2020gong}:
\begin{equation}
    \Omega=\sum_{g_l,g_r \in G_0} \omega(g_r,g_r^{-1}g_l)\ket{g_l,g_r}\bra{g_l,g_r}, 
\end{equation}
where $\omega$ is a 2-cocycle with nontrivial cohomology class. Now we have $[\mathbb{1}_L\otimes\mathcal{D}(g),\rho_a]=0$, where the projective representation $\mathcal{D}(g)$ identified from $\mathcal{D}(g)\otimes\overline{\mathcal{D}(g)}=\Omega \left[D(g) \otimes D(g)\right] \Omega^\dagger$ satisfies $\mathcal{D}(g)\mathcal{D}(h)=\omega(g,h)\mathcal{D}(gh)$ $\forall g,h\in G_0$.
In the presence of TRS, one can separate the solution into $\omega(t,t)=\pm1$ and the unitary part by an appropriate gauge fixing, with the latter further constrained by $\omega(g,h)\omega(\tilde{g},\tilde{h})=1$ $\forall g,h\in G_0$ on top of the (unitary) cocycle condition \cite{yang2017irreducible}. If $\omega(t,t)=-1$, one should further apply $\Upsilon$ gate (\ref{Ups}) to two 2-dimensional subsystems $\sigma_L,\sigma_R=\pm1$ in $L=L'\sigma_L$ and $R=\sigma_R R'$ of $\Omega\ket{\Psi}$ to fractionalize the TRS. See Fig.~\ref{fig:setup}(c) for an illustration for the most general case.

In summary, $\rho_a$ obtained from $G$-symmetric random states with possible symmetry fractionalization is determined by 
\begin{equation}
w_{ab}=\frac{1}{\sqrt{|G_0|}}\omega(g_r,g_r^{-1}g_l)c_{L,g_r^{-1}g_l,R},
\label{wab}
\end{equation}
where $\omega$ is a unitary 2-cocycle on $G_0$. In the presence of TRS, we further have $c_{L,\tilde g,R}=\overline{c_{L,g,R}}$ ($c_{L',\tilde g,R'}=\sigma_y\overline{c_{L',g,R'}}\sigma_y$ in the $2$-by-$2$ block representation of $c_{L,g,R}=[c_{L',g,R'}]_{\sigma_L,\sigma_R}$) if $\omega(t,t)=1$ ($\omega(t,t)=-1$), and $\omega(g,h)\omega(\tilde g,\tilde h)=1$.

\emph{Universal decomposition into the threefold way.---}
It remains to understand the eigenvalue statistics of $G$-symmetric $\rho_a$, or equivalently, the singular value statistics of Eq.~(\ref{wab}). This turns out to be systematically solvable using the (generalized) group representation theory \cite{yang2017irreducible}. A fundamental fact is that any (projective) representation $\mathcal{D}(g)$ can be decomposed into $\bigoplus_\alpha \mathcal{D}^\alpha(g)$, a direct sum of irreducible representations (irreps). For each irrep $\alpha$, one can define an \emph{indicator} \cite{kawanaka1990twisted,bradley2009mathematical}
\begin{equation}
\label{indmt}
    \iota_\alpha=\frac{1}{|G_0|}\sum_{g\in G_0} \omega(\tilde{g},g)\chi_\alpha(\tilde{g}g),
\end{equation}
where $\chi_\alpha(g)=\Tr \mathcal{D}^\alpha(g)$. In particular, $d_\alpha=\chi_\alpha(e)$ is the dimension of irrep $\alpha$. This indicator (\ref{indmt}) takes values on $0,1,-1$, corresponding to three different possibilities about how $\overline{\mathcal{D}^\alpha(\tilde g)}$ is related to $\mathcal{D}^\alpha(g)$. If $\iota_\alpha=0$, they are different irreps. Otherwise, they are equivalent irreps and one can make 
$\overline{\mathcal{D}^\alpha(\tilde g)}=\mathcal{D}^\alpha(g)$ ($(\sigma_y\otimes\mathbb{1}_{d_\alpha/2})\overline{\mathcal{D}^\alpha(\tilde g)}(\sigma_y\otimes\mathbb{1}_{d_\alpha/2})=\mathcal{D}^\alpha(g)$) if $\iota_\alpha=1$ ($\iota_\alpha=-1$).
Having these preliminaries in mind, we introduce the following theorem \cite{SM}:
\begin{theorem}
\label{thm:espectrum}
    The matrix ensemble given in Eq.~(\ref{wab}) is completely decomposed into the direct sum of the threefold way.
    If $G=G_0$, the matrix ensemble of $WW^\dag$ is
    \begin{equation}
        \bigoplus_{\alpha}
        \frac{\mathbb{1}_{d_\alpha}}{d_\alpha} \otimes \mathrm{LUE}^{d_Ld_\alpha\times d_Rd_\alpha}_\alpha. \label{eq:onefold}
    \end{equation}
    Otherwise $G=G_0\rtimes \mathbb{Z}_2^\mathcal{T}$, the ensemble of $WW^\dag$ is 
    \begin{gather}
        \left[\bigoplus_{\alpha:R_{+}}\frac{\mathbb{1}_{d_\alpha}}{d_\alpha} \otimes \mathrm{LOE}_\alpha^{d_Ld_\alpha\times d_Rd_\alpha}\right]\nonumber\\
        \oplus\nonumber\\
        \left[\bigoplus_{\alpha:R_0}\frac{\mathbb{1}_{d_\alpha}}{d_\alpha} \otimes \left(\mathrm{LUE}_\alpha^{d_Ld_\alpha\times d_Rd_\alpha}\oplus \overline{\mathrm{LUE}_\alpha^{d_Ld_\alpha\times d_Rd_\alpha}}\right)\right]\nonumber\\
        \oplus\nonumber\\
        \left[\bigoplus_{\alpha:R_{-}}\frac{\mathbb{1}_{d_\alpha}}{d_\alpha} \otimes \mathrm{LSE}_\alpha^{d_Ld_\alpha\times d_Rd_\alpha}\right], \label{eq:threefold}
    \end{gather}
    where $R_\pm$ is the set of (projective) irreps $\mathcal{D}^\alpha$ satisfying $\iota_\alpha=\pm\omega(t,t)$, 
    while $R_0$ includes only one components ($\alpha$) of the involution pairs ($\{\alpha,\alpha^\star\}$) with $\iota_\alpha=0$. Here the subscript $\alpha$ labels an irrep of $G_0$ with dimension $d_\alpha$.
    The superscript $d_Ld_\alpha\times d_Rd_\alpha$ denotes the size of the matrix block, and the blocks labeled by different $\alpha$ are independent.
\end{theorem}
We emphasize that, given $G$ and $d_{L,R}$, this decomposition depends only on the cohomology class of the cocycle $\omega$. 
Remarkably, there are only three building blocks, i.e., LOE, LUE, and LSE, which parallel with the Gaussian ensembles in the Hamiltonian counterpart \cite{dyson1962threefold}.
The block decomposition in Eqs.~(\ref{eq:onefold}) and (\ref{eq:threefold}) is reminiscent of symmetry-resolved entanglement \cite{Goldstein2018,Azses2020,Horvath2020,Kusuki2023,Saura2024categorical}. 
In this context, our result provides a possible generalization to include anomalous symmetries, and makes a systematic connection to RMT.
Note that symmetry fractionalziation by $\Omega$ and/or $\Upsilon$ causes a change on the dimensions and types of irreps in general.
Actually there is a limitation on the change of matrix ensemble \cite{SM}:
In the absence (presence) of $\Upsilon$, the difference of the number of LOEs blocks and that of LSEs blocks, including degeneracy, never increases (decreases) upon the action of $\Omega$. 

\begin{table}[tbp]
     \centering
     \caption{Matrix ensembles for $G_0=\mathbb{Z}_2$. Blocks labeled by different subscripts are statistically independent.}
     \label{tbl:Z2}
     \begin{tabular}{c|c|c|c|c}
         \;\;$\mathbb{Z}_2^{\mathcal{T}}$\;\; & \;$\omega(p,p)$\; & \;$\omega(t,t)$\; & \;\;\;\;\;\;\;\;Ensemble\;\;\;\;\;\;\;\; & \;Degeneracy\;
         \\
         \hline \hline
         $\times$ & N/A & N/A & $\mathrm{LUE}_1\oplus \mathrm{LUE}_2$ & 1 \\ 
         $\checkmark$ & $+$ & $+$ & $\mathrm{LOE}_1\oplus \mathrm{LOE}_2$ & 1 \\ 
         $\checkmark$ & $-$ & $+$ & $\mathrm{LUE}\oplus \overline{\mathrm{LUE}}$ & 2 \\ 
         $\checkmark$ & $+$ & $-$ & $\mathrm{LSE}_1\oplus \mathrm{LSE}_2$ & 2\\
         $\checkmark$ & $-$ & $-$ & $\mathrm{LUE}\oplus \overline{\mathrm{LUE}}$ & 2 \\ 
     \end{tabular}
\end{table}

\emph{Examples.---}Clearly LUE, LOE, LSE are repdocuded by taking $G_0=\{e\}$ and $\omega(t,t)=\pm1$ (if TRS is imposed). The next minimal example is $G_0=\mathbb{Z}_2=\{e,p\}$. Without TRS, symmetry fractionalization never occurs and the matrix ensemble is a direct sum of two independent LUEs. In the presence of TRS, taking $g=h=p$ in $\omega(g,h)\omega(\tilde g,\tilde h)=1$, we obtain inequivalent cocyles $\omega(p,p)=\pm1$ on top of $\omega(t,t)=\pm1$ \cite{yang2017irreducible}. The ensemble turns out to be a pair of complex conjugate LUEs whenever $\omega(p,p)=-1$, and otherwise two independent LOEs ($\omega(t,t)=1$) or LSEs ($\omega(t,t)=-1$) if $\omega(p,p)=1$. See Table~\ref{tbl:Z2} for a summary of the results.

Note that entanglement-spectrum degeneracy in the above example arises from symmetry fractionalization. 
On the other hand, degeneracy appears for any non-Abelian group even in the absence of symmetry fractionalization. The simplest non-Abelian group is $G=C_{3v}$ generated by $\frac{2\pi}{3}$ rotation and mirror operation. It has three linear irreps: two of them are 1-dimensional, and the rest one is 2-dimensional. Hence, without TRS, the matrix ensemble reads 
\begin{equation}
    \mathrm{LUE}_1^{d_L\times d_R}\;\oplus\; \mathrm{LUE}_2^{d_L\times d_R}\;\oplus\; \frac{\mathbb{1}_2}{2}\otimes\mathrm{LUE}_3^{2d_L\times 2d_R}.
\end{equation}
One can also obtain LSE blocks without fractionalizing the TRS (i.e., $\omega(t,t)=1$), as exemplified by $G=Q_8\times\mathbb{Z}_2^{\mathcal{T}}$. Here $G_0=Q_8$ is the quaternion group, which has four 1-dimensional irreps and one 2-dimensional irrep. The former irreps have indicator $\iota=1$ (cf. Eq.~(\ref{indmt})). The latter irrep is the spin representation of quaternion and has indicator $\iota=-1$. The matrix ensemble is thus given by 
\begin{equation}
    \left(\bigoplus_{\alpha=1}^4\mathrm{LOE}_\alpha^{d_L\times d_R}\right)\;\oplus\; \frac{\mathbb{1}_2}{2}\otimes \mathrm{LSE}
    ^{2d_L\times 2d_R}.
\end{equation}

\emph{Summary and outlook.---}
By incorporating symmetry fractionalization into RMT, we succeeded in constructing the LSE entanglement spectrum. 
We figured out the explicit forms of random matrices describing the entanglement spectra of symmetric random states. Moreover, we showed the matrix ensemble can always be decomposed into a direct sum of LUE, LOE, and/or LSE. 
Our results can be interpreted as the Laguerre version of the Dyson's threefold way \cite{dyson1962threefold}. 

Finally, we would like to discuss some future prospects. 
In our work, only $0$-form invertible symmetry described by finite groups are considered.
It is natural to ask how our conclusions may be changed by continuous symmetries described by Lie groups, and more advanced higher-form and/or non-invertible symmetries \cite{gaiotto2015generalized,mcgreevy2023generalized,choi2023remarks,ji2020categorical}. 
Another direction is to consider entanglement under superselection rules \cite{schuch2004nonlocal,schuch2004quantum,you2017sachdev-ye-kitaev}, such as fermionic systems with definite fermion-number parities \cite{MariCarmen2007}. In fermionic systems, the anti-unitary TRS is no longer involutory \cite{Fidkowski10,Fidkowski11,yang2017irreducible}, and there could be more different types of fundamental random matrix ensembles \cite{altland1997nonstandard,kawabata2024symmetry}. 

We are grateful to Lucas Hackl, Haruki Watanabe, and Takahiro Morimoto for valuable discussions. 
H.Y. acknowledges support from FoPM, WINGS Program, the University of Tokyo. 
K.M. was supported by JSPS KAKENHI Grant. No. JP23K13037. 
Z.G. acknowledges support from the University of Tokyo Excellent Young Researcher Program and from JST ERATO Grant Number JPMJER2302, Japan.

\bibliography{cite}

\clearpage
\begin{center}
\textbf{\large Supplemental Materials}
\end{center}
\setcounter{equation}{0}
\setcounter{figure}{0}
\setcounter{table}{0}
\setcounter{theorem}{0}
\makeatletter
\renewcommand{\theequation}{S\arabic{equation}}
\renewcommand{\thefigure}{S\arabic{figure}}
\renewcommand{\bibnumfmt}[1]{[S#1]}
\renewcommand{\thetheorem}{S\arabic{theorem}}

We review the basics of projective representations of finite groups. 
We then provide the details on the proof of Theorem~\ref{thm:espectrum} in the main text, including some useful lemmas.

\section{Definitions}

In order to make this material self-contained, let us briefly review the definitions. 
Suppose a finite group $G$ describes the symmetry of the system.
We consider the case $G=G_0$ or the case $G=G_0\rtimes \mathbb{Z}_2^\mathcal{T}$.
Every group element $g \in G_0$ is represented by an unitary matrix. 
In the presence of TRS $\mathbb{Z}_2^\mathcal{T}=\{e,t\}$, the representation of $gt$ is anti-unitary $\forall g\in G_0$. 
Note that the transformation by $t \in \mathbb{Z}_2^\mathcal{T}$ is an involutory automorphism on $G_0$, namely $t(\cdot)t :G_0 \to G_0$, then it is convenient to define $\forall g \in G_0, t g t=:\tilde{g} \in G_0$. 
The Hilbert space of a $G_0$-spin is spanned by $\{|g\rangle\}_{g\in G_0}$. Fixing the action of $t$ to $\ket{e}$ to satisfy $D(t)\ket{e}=\ket{e}$, we get $D(t)\ket{g}=\ket{\tilde{g}}$ from $D(t)D(g)D(t)=D(\tilde{g})$, where $D(g)$ ($g\in G_0$) is assumed to be the regular representation. 
Concretely, we have
\begin{equation}
\begin{split}
&D(g)=\sum_{h\in G_0} |gh\rangle\langle h|,\;\;\forall g\in G_0,\\
&D(t)=\sum_{g\in G_0} |\tilde g\rangle\langle g| K,
\end{split}
\end{equation}
where $K$ denotes the complex conjugate operation under the basis $\{|g\rangle\}_{g\in G_0}$.

\subsection{Projective representations of groups}

In this section, we explain projective representation of groups, theorems related to that, 
and its application to our results.
A projective representation of a group $G$ with 2-cocycle $\omega:G\times G\to U(1)$ is a set of unitary and anti-unitary operators $\{\mathcal{D}(g)\}_{g\in G}$ which satisfy
\begin{equation}
    \mathcal{D}(g)\mathcal{D}(h)=\omega(g,h)\mathcal{D}(gh).
\end{equation}
A representation of $\omega=1$ is called a linear (anti-linear) representation. 
From the request of associativity, 2-cocycle $\omega$ satisfies $\forall g,g',g''$,
\begin{equation}
\omega(g,g')\omega(gg',g'')=\omega(g,g'g'')\omega^g(g',g''), 
\end{equation}
where $\omega^g=\omega$ if the representation of $g$ is unitary, otherwise $\omega^g=\overline{\omega}=\omega^{-1}$.

Redefining the representation with a phase modulation $\beta$: $G\to U(1)$ via $\mathcal{D}'(g)=\beta(g)\mathcal{D}(g)$ produces a new 2-cocycle 
\begin{equation}
\label{eq:coboundary}
    \omega'(g,h)=\frac{\beta(gh)}{\beta(g)\beta^g(h)}\omega(g,h),
\end{equation}
where $\beta^g=\beta$ if the representation of $g$ is unitary, otherwise $\beta^g=\overline{\beta}=\beta^{-1}$.
The redefined projective representation satisfies $\mathcal{D}'(g)\mathcal{D}'(h)=\omega'(g,h)\mathcal{D}'(gh)$. 
This phase modulation (coboundary) induces equivalence classes of $\omega$, and they are classified by the second-order group cohomology $H^2(G,U(1))$ \cite{chen2013symmetry}. These equivalence classes are called cohomology classes.  

Remember that we focus on the case $G=G_0\rtimes \mathbb{Z}_2^\mathcal{T}$ for the anti-unitary case. This assumption allows us to decouple $G_0$ and $\mathbb{Z}_2^\mathcal{T}$ parts in $\omega$. We can fix the gauge that separates the cocycles into $G_0$ elements and the $\mathbb{Z}_2^\mathcal{T}$ parts \cite{yang2017irreducible}:
\begin{equation}
\label{eq:decouplings}
\begin{split}
    \omega(e,g)&=\omega(g,e)=1,\\
    \omega(t,t)&=\pm 1,\\
    \omega(g,t)&=\omega(t,g)=1,\\
    \omega(gt,h)&=\omega(g,\tilde{h}),
    \\
    \omega(g,ht)&=\omega(g,h),\\
    \omega(gt,ht)&=\omega(g,\tilde h)\omega(t,t),
\end{split}
\end{equation}
where $g,h\in G_0$, $\tilde{g}=tgt\in G_0$. Here $\omega(g,h)$ with $g,h\in G_0$ is a unitary $2$-cocycle subject to additional constraint $\omega(g,h)\omega(\tilde g,\tilde h)=1$. Also, a phase modulation under the above gauge fixing is constrainted by $\beta(e)=1$, $\beta(gt)=\beta(g)\beta(t)$ and $\beta(g)\beta(\tilde g)=1$. 
Similarly, any unitary and antiunitary parts of $\mathcal{D}(g)$ can be considered separately. From definition, any $g\in G\setminus G_0$ can be represented by $\mathcal{D}(g)=\mathcal{D}(g_0)D(t)$.

It should be noted that nontrivial cohomology classes can arise from the interplay between $G_0$ and TRS. 
For example, $G_0=\mathbb{Z}_2$ itself has only one (trivial) cohomology class. However, if TRS is assumed, $G=G_0\times \mathbb{Z}_2^\mathcal{T}$ has four cohomology classes \cite{yang2017irreducible}, even though $\mathbb{Z}_2^\mathcal{T}$ has only 2 cohomology classes. As mentioned in the main text, the new nontrivial class arises from the inequivalence between $\omega(p,p)=\pm1$, which arises from taking $g=h=p$ in $\omega(g,h)\omega(\tilde g,\tilde h)=1$. 
This 2-cocycle itself is always trivial without TRS. 

On the other hand, it is also possible that the group cohomology of $G$ is trivial despite that of $G_0$ is nontrivial. A simple example is $G=G_0\times \mathbb{Z}_2^\mathcal{T}$ with $G_0=\mathbb{Z}_3\times\mathbb{Z}_3$ (though the situation changes for $G_0\rtimes\mathbb{Z}_2^\mathcal{T}$ \cite{yang2017irreducible}). Suppose $\omega$ is a nontrivial $2$-cocycle of $G_0$, we know that $\omega^3$ is trivial since $H^2(G_0,U(1))=\mathbb{Z}_3$. However, we also have $\omega^2=1$ due to $\omega(g,h)\omega(\tilde g,\tilde h)=1$ and $\tilde g=g$. This implies $\omega=\omega^3$ is trivial as well.
Note that here the triviality is in the sense of $H^2(G_0,U(1))$. It remains unclear whether there exists a cocycle that was trivial without TRS but becomes nontrivial with TRS. 
It turns out there is no such a cocycle, as shown below: 
Let $\forall g,h\in G_0$, then $g^3=h^3=e$ and $gh=hg$ are satisfied. 
First we check $\omega(g,g)$ can be gauged to be 1 under the constraint of phase modulations: $\beta(e)=1,\,\beta(gt)=\beta(g)\beta(t),$ and $\beta(g)\beta(\tilde{g})=1$. The last constraint becomes $\beta(g)^2=1$ because now we consider $G=G_0\times\mathbb{Z}_2^\mathcal{T}$. From Eq.~(\ref{eq:coboundary}), a possible phase modulation of $\omega(g,g)$ is
\begin{equation}
    \omega'(g,g)=\frac{\beta(g^2)}{\beta(g)^2}\omega(g,g)=\beta(g^2)\omega(g,g).
\end{equation}
Since $\omega(g,g)$ and $\beta(g^2)=\beta(g^{-1})$ have the same degree of freedom $|G_0|$, we can always fix $\beta(g^2)=\omega^{-1}(g,g)$. Thus, $\forall g\in G_0, \omega(g,g)$ can be gauge transformed into $\omega(g,g)=1$.
Next we check $\omega(g,h)=1$ also holds for $\forall g,h\in G_0$. Since $\omega(g,h)$ is trivial without TRS, $\omega(g,h)$ satisfies the coboundary condition $\omega(g,h)=\frac{\gamma(gh)}{\gamma(g)\gamma(h)}$ for some $\gamma:G\to U(1)$. Using $\omega(g^2,g^2)=\omega(h^2,h^2)=\omega((gh)^2,(gh)^2)=1$ and $\omega(g^2,h^2)^2=1$, we get 
\begin{equation}
\begin{split}
    \omega(g,h)&=\frac{\omega(g,h)}{\omega(g^2,h^2)^2}\\
    &=\frac{\gamma(gh)}{\gamma(g)\gamma(h)}\frac{\gamma(g^2)^2\gamma(h^2)^2}{\gamma((gh)^2)^2}\\
    &=\frac{\gamma(gh)}{\gamma((gh)^2)^2}\frac{\gamma(g^2)^2}{\gamma(g)}\frac{\gamma(h^2)^2}{\gamma(h)}\\
    &=\frac{\omega((gh)^2,(gh)^2)}{\omega(g^2,g^2)\omega(h^2,h^2)}\\
    &=\frac{1}{1\cdot 1}=1.
\end{split}
\end{equation}
This result suggests that any cocycle of $G_0=\mathbb{Z}_3\times\mathbb{Z}_3$ which is trivial without TRS is also trivial with TRS. Therefore, $G=G_0\times\mathbb{Z}_2^\mathcal{T}$ has no nontrivial cohomology classes which were trivial without $\mathbb{Z}_2^\mathcal{T}$. Thus $H^2(G,U_{\mathcal{T}}(1))=\mathbb{Z}_2$ is proven. We expect the above proof can be generalized to any Abelian group with odd $|G_0|$.

\subsection{Theorems}

Some theorems similar to the linear representation cases hold even if the representation is projective.
\begin{theorem}
\label{thm:orthogonality}
    Let $\mathcal{D}^\alpha$ be a unitary projective irrep of a group $G_0$ (not $G$). The grand orthogonality theorem
    \begin{equation}
    \label{eq:orthogonality}
        \sum_{g\in G_0}\mathcal{D}^\alpha_{ij}(g)\overline{\mathcal{D}^\beta_{kl}(g)}=\frac{|G_0|}{d_\alpha}\delta_{\alpha,\beta}\delta_{i,k}\delta_{j,l}
    \end{equation}
    and the dimension equation
    \begin{equation}
    \label{eq:dimensioneq}
        \sum_\alpha d_\alpha^2=|G_0|
    \end{equation}
    hold as in the case of linear representations.
\end{theorem}
\begin{proof}
    We make use of the generalized Schur's lemma \cite{yang2017irreducible}:
    \begin{gather}
        \forall g\in G_0, \, \mathcal{D}^\alpha(g)X=X\mathcal{D}^\beta(g)\nonumber\\
        \Rightarrow 
        \left\{ \,
            \begin{aligned}
            & X\propto \mathbb{1}_{d_\alpha}\, & (\alpha=\beta) \\
            & X=0\, & (\alpha\neq\beta)
            \end{aligned}
        \right. .
        \label{Schur}
    \end{gather}
    An example of $X$ is
    \begin{equation}
        X=\sum_{h\in G_0}\frac{1}{\omega(h^{-1},h)}\mathcal{D}^\alpha(h)A\mathcal{D}^\beta(h^{-1}),
    \end{equation}
    where $A$ is an arbitrary $d_\alpha \times d_\beta$ matrix. To see this, we only have to check that $\forall g\in G_0$,
    \begin{equation}
        \begin{split}
            &\mathcal{D}^\alpha(g)X\\
            &=\sum_{h\in G_0}\frac{1}{\omega(h^{-1},h)}\mathcal{D}^\alpha(g)\mathcal{D}^\alpha(h)A\mathcal{D}^\beta(h^{-1})\\
            &=\sum_{h\in G_0}\frac{\omega(g,h)}{\omega(h^{-1},h)}\mathcal{D}^\alpha(gh)A\mathcal{D}^\beta(h^{-1})\\
            &=\sum_{h'\in G_0}\frac{\omega(g,g^{-1}h')}{\omega(h'^{-1}g,g^{-1}h')}\mathcal{D}^\alpha(h')A\mathcal{D}^\beta(h'^{-1}g)\\
            &=\sum_{h'\in G_0}\frac{\omega(g,g^{-1}h')\mathcal{D}^\alpha(h')A\mathcal{D}^\beta(h'^{-1})}{\omega(h'^{-1}g,g^{-1}h')\omega(h'^{-1},g)}\mathcal{D}^\beta(g)\\
            &=\sum_{h'\in G_0}\frac{\omega(g,g^{-1}h')\mathcal{D}^\alpha(h')A\mathcal{D}^\beta(h'^{-1})}{\omega(h'^{-1},h')\omega(g,g^{-1}h')}\mathcal{D}^\beta(g)\\
            &=X\mathcal{D}^\beta(g).
        \end{split}
    \end{equation}
    Note that $X$ can be rewritten as
    \begin{equation}
        \begin{split}
            X&=\sum_{h\in G_0}\mathcal{D}^\alpha(h)A{\mathcal{D}^\beta}(h)^\dagger.
        \end{split}
    \end{equation}
    Since $A$ is arbitrary, we may choose $A=\ket{j}\bra{l}$. From the generalized Schur's lemma (\ref{Schur}), we have
    \begin{equation}
        X_{ik;jl}=\sum_{h\in G_0}\mathcal{D}_{ij}^\alpha(h)\overline{\mathcal{D}^\beta_{kl}(h)}=\lambda\delta_{\alpha,\beta}\delta_{i,k}.
    \end{equation}
    Here constant $\lambda$ can be obtained by summing $i$ up under the assumption $i=k$ and $\alpha=\beta$, namely
    \begin{equation}
    \begin{split}
        \sum_{i}
        X_{ii;jl}&=\left[\sum_{h\in G_0}{\mathcal{D}^\alpha}(h)^\dagger\mathcal{D}^\alpha(h)\right]_{lj}
        \\
        &=\left[\sum_{h\in G_0}\mathbb{1}_{d_\alpha}\right]_{lj}
        \\
        &=|G_0|\delta_{j,l}.
    \end{split}
    \end{equation}
    On the other hand, $\sum_i 
        X_{ii;jl}=d_\alpha\lambda$. 
    By comparing these two results, we get $\lambda=\frac{|G_0|}{d_\alpha}\delta_{j,l}$, implying Eq.~(\ref{eq:orthogonality}). 

    As for Eq.~(\ref{eq:dimensioneq}), we consider the projective regular representation 
    \begin{equation}
    \label{eq:projregular}
    \mathcal{D}(g)=\sum_{g'}\omega(g,g')\ket{g g'}\bra{g'},
    \end{equation}
    whose character (trace) turns out to be 
    \begin{equation}
        \begin{split}
            \chi_\mathrm{reg}(g)&:=\Tr{\mathcal{D}(g)}\\
            &=\Tr{\sum_{g'\in G_0} \omega(g,g')\ket{gg'}\bra{g'}}\\
            &=\sum_{g'\in G_0}\omega(g,g')\langle g'|gg'\rangle \\
            &=\left\{ \,
            \begin{aligned}
            & |G_0|\, & (g=e) \\
            & 0\, & (g\neq e)
            \end{aligned}
        \right. .
        \end{split}
    \end{equation}
    On the other hand, since $\mathcal{D}(g)$ is reducible, we get
    \begin{equation}
    \label{eq:charactersum}
        \begin{split}
            \Tr{\mathcal{D}}&=\Tr\left[{\bigoplus_\alpha \mathbb{1}_{q_\alpha}\otimes \mathcal{D}^\alpha}\right]\\
            \Rightarrow|G_0|\delta_{g,e}&=\sum_\alpha q_\alpha\chi_\alpha(g).
        \end{split}
    \end{equation}
    From Eq.~(\ref{eq:orthogonality}), 
    we can derive the orthogonality of characters:
    \begin{equation}
    \begin{split}
        &\sum_{g\in G_0}\chi_\alpha(g)\overline{\chi_\beta(g)}\\
        &=\sum_{i,k}\sum_{g\in G_0}\mathcal{D}^\alpha_{ii}(g)\overline{\mathcal{D}^\beta_{kk}(g)}\\
        &=\sum_{i,k}\frac{|G_0|}{d_\alpha}\delta_{\alpha,\beta}\delta_{i,k}\\
        &=|G_0|\delta_{\alpha,\beta}.
    \end{split}
    \end{equation}
    By utilizing this orthogonality, we can identify $q_\alpha$:
    \begin{equation}
        \begin{split}
            \sum_{g\in G_0}\delta_{g,e}\overline{\chi_\beta(g)}&=\frac{1}{|G_0|}\sum_{g\in G_0}\sum_\alpha q_\alpha\chi_\alpha(g)\overline{\chi_\beta(g)}\\
            \Rightarrow d_\beta&=q_\beta\\
        \end{split}
    \end{equation}
    Substituting $g=e$ into Eq.~(\ref{eq:charactersum}), we get Eq.~(\ref{eq:dimensioneq}), and Theorem~\ref{thm:orthogonality} is proven.
\end{proof}

\section{Symmetry concentration}
In the main text, we used the fact that a unitary transformation can concentrate 
the on-site symmetry action $\{D(g)^{\otimes N}\}$ to only one site $\{\mathbb{1}_{|G_0|^{N-1}}\otimes D(g)\}$ (this is applied to both subsystems). 
The feasibility in the unitary case ($G=G_0$) can be understood from the character theory. On the other hand, it is not clear whether this is still true in the presence of TRS, and how one can explicitly execute the concentration procedure.  
We answer these questions in the following. 

Since the concentration can be done step by step, it suffices to identify a two-site unitary transformation $u$ that satisfies $u \left(D(g)\otimes D(g)\right) u^\dagger=\mathbb{1}_{|G_0|}\otimes D(g)$. 
The following $u$ is one suitable choice:
\begin{equation}
    \begin{split}
        u&=U_2 U_1\\
        U_1&=\sum_{g\in G_0}D(g)\otimes \ket{g^{-1}}\bra{g^{-1}}\\
        U_2&=u_2\otimes \mathbb{1}_{|G_0|}\\
        u_2&=\sum_{g\in G_0}\frac{1}{\sqrt{2}}(\omega_8|g\rangle + \omega_8^{-1}|\tilde g\rangle)\langle g|\\
        \omega_8&=e^{i\frac{\pi}{4}}.
    \end{split}
\end{equation}
One can check that $u$ is unitary;
\begin{equation}
\begin{split}
    U_1 U_1^\dag&=\sum_{g,g'\in G_0}D(g)D(g'^{-1})\otimes \ket{g^{-1}}\braket{g^{-1}|g'^{-1}}\bra{g'^{-1}}\\
    &=\sum_{g\in G_0}D(e)\otimes \ket{g^{-1}}\bra{g^{-1}}\\
    &=\mathbb{1}_{|G_0|}\otimes\mathbb{1}_{|G_0|},\\
    U_2 U_2^\dag&=u_2 u_2^\dag \otimes \mathbb{1}_{|G_0|},\\
    u_2 u_2^\dag&=\frac{1}{2}\sum_{g\in G_0}\left(\omega_8 \ket{g}+\omega_8^{-1}\ket{\tilde{g}}\right)\left(\omega_8^{-1} \bra{g}+\omega_8\bra{\tilde{g}}\right)\\
    &=\frac{1}{2}\sum_{g\in G_0}\left(\ket{g}\bra{g}+\ket{\tilde{g}}\bra{\tilde{g}}+i\ket{g}\bra{\tilde{g}}-i\ket{\tilde{g}}\bra{g}\right)\\
    &=\mathbb{1}_{|G_0|},
\end{split}
\end{equation}
and satisfies $u \left(D(g)\otimes D(g)\right) u^\dagger=\mathbb{1}_{|G_0|}\otimes D(g)$. In case of $g\in G_0$, we have
\begin{equation}
    \begin{split}
        &U_1 \left(D(g)\otimes D(g)\right) U_1^\dag\\
        &=\sum_{g',g''\in G_0}D(g'gg''^{-1})\otimes \ket{g'^{-1}}\braket{g'^{-1}|g g''^{-1}}\bra{g''^{-1}}\\
        &=\sum_{g''\in G_0}D(g''g^{-1} g g''^{-1})\otimes \ket{g g''^{-1}}\bra{g''^{-1}}\\
        &=\mathbb{1}_{|G_0|}\otimes \sum_{g''\in G_0}D(g) \ket{g''}\bra{g''}\\
        &=\mathbb{1}_{|G_0|}\otimes D(g),\\
        &U_2\left(\mathbb{1}_{|G_0|}\otimes D(g)\right)U_2^\dag\\
        &=\left(u_2\mathbb{1}_{|G_0|}u_2^\dag\right)\otimes D(g)\\
        &=\mathbb{1}_{|G_0|}\otimes D(g).
    \end{split}
\end{equation}
Otherwise, $D(g)$ can be written as $D(g_0)D(t)$ for $(g_0,t)\in G_0\rtimes \mathbb{Z}_2^\mathcal{T}$, so that
\begin{equation}
    \begin{split}
        &U_1 \left(D(g_0)D(t)\otimes D(g_0)D(t)\right) U_1^\dag\\
        &=\sum_{g',g''\in G_0}D(g'g_0\widetilde{g''^{-1}}t)\otimes \ket{g'^{-1}}\braket{g'^{-1}|g_0 \widetilde{g''^{-1}}}\bra{g''^{-1}}\\
        &=\sum_{g'\in G_0}D(g'g_0g_0^{-1}g'^{-1}t)\otimes \ket{g'^{-1}}\bra{\widetilde{g_0^{-1}g'^{-1}}}\\
        &=D(t)\otimes \sum_{g'\in G_0} \ket{g'}\bra{g'}D(g_0)D(t)\\
        &=D(t)\otimes D(g),
    \end{split}
\end{equation}
\begin{equation}
    \begin{split}
        &U_2\left(D(t)\otimes D(g)\right)U_2^\dag\\
        &=u_2D(t)u_2^\dag \otimes D(g)\\
        &=\frac{1}{2}\sum_{g\in G_0}\left(\omega_8\ket{g}+\omega_8^{-1}\ket{\tilde{g}}\right)\left(\omega_8\bra{\tilde{g}}+\omega_8^{-1}\bra{g}\right) \otimes D(g)\\
        &=\mathbb{1}_{|G_0|} \otimes D(g).
    \end{split}
\end{equation}
Together, the above equations show that $u$ condenses 2-site symmetry to 1-site symmetry.
Continuing application of $u$ until the number of $D(g)$ reaches one, the equivalence of $D(g)^{\otimes N}$ and $\mathbb{1}_{|G_0|^{N-1}}\otimes D(g)$ is proven.

\section{Setup and corollaries}

Recall that our rearranged system has four subsystems $L,l,r$, and $R$ from left to right, which undergo a symmetry transformation of $\mathbb{1}_L,D(g),D(g)$, and $\mathbb{1}_R$, respectively. 
Within the symmetric subspace, subsystem $l\cup r$ has $|G_0|$ independent basis states: 
\begin{equation}
    \ket{\psi_g} = \frac{1}{\sqrt{|G_0|}} \sum_{h \in G_0} \ket{h g}\ket{h}, \quad \forall g \in G_0.
\end{equation}
To see this, we recall the action of $g'\in G$ is $D(g')\otimes D(g')$ on $l\cup r$, thus $\forall g,g'\in G_0$ we have
    \begin{equation}
    \begin{split}
        D(g')\otimes D(g')\ket{\psi_g}&=\frac{1}{\sqrt{|G_0|}} \sum_{h \in G_0} \ket{g'h g}\ket{g'h}\\
        &=\frac{1}{\sqrt{|G_0|}} \sum_{h \in G_0} \ket{h g}\ket{h}=\ket{\psi_g}.
    \end{split}
    \end{equation}
    In contrast, the action of time reversal switches $\ket{\psi_g}\leftrightarrow \ket{\psi_{\tilde g}}$, 
since
    \begin{equation}
    \begin{split}
        D(t)\otimes D(t)
        \ket{\psi_g}&=\frac{1}{\sqrt{|G_0|}} \sum_{h \in G_0} \ket{\tilde{h g}}\ket{\tilde{h}}\\
        &=\frac{1}{\sqrt{|G_0|}} \sum_{h \in G_0} \ket{h \tilde{g}}\ket{h}=\ket{\psi_{\tilde{g}}}.
    \end{split}
    \end{equation}
A general $G$-invariant global state is then given by
\begin{equation}
\begin{split}
    \ket{\Psi} &= \sum_{L,g,R} c_{L,g,R} \ket{L}\ket{\psi_g}\ket{R}\\
    &= \frac{1}{\sqrt{|G_0|}} \sum_{L,g_l,g_r,R} c_{L,g_r^{-1}g_l, R}\ket{L}\ket{g_l}\ket{g_r}\ket{R},
\end{split}
\label{SMPsi}
\end{equation}
where $c_{L,g,R}=\overline{c_{L,\tilde g,R}}$ in the presence of TRS.

In this setup, the density matrix $\rho=\Tr_{rR}|\Psi\rangle\langle\Psi|$ is invariant under the action of $\mathbb{1}_L \otimes D(g)$ for $\forall g \in G$, i.e., $[\rho,\mathbb{1}_L\otimes D(g)]=0$. To incorporate symmetry fractionalization specified by a nontrivial 2-cocycle $\omega$, one can apply the following local unitary gate (supported on $l\cup r$) $\Omega$ to $\ket{\Psi}$:
\begin{equation}
    \Omega=\sum_{g_l,g_r\in G_0} \omega(g_r,g_r^{-1}g_l)\ket{g_l,g_r}\bra{g_l,g_r}. 
\end{equation}
The representation $\mathcal{D}(g)$ identified from  
$
\mathcal{D}(g)\otimes\overline{\mathcal{D}(g)}
=\Omega \left[D(g) \otimes D(g) \right] \Omega^\dagger$ is projective and satisfies $\mathcal{D}(g)\mathcal{D}(h)=\omega(g,h)\mathcal{D}(gh)$.
To see this, we first calculate the $\Omega$-transformation of $D(g)\otimes D(g)$.
As for a unitary element $g\in G_0$, we have
\begin{equation}
\begin{split}
    &\Omega \left[D(g)\otimes D(g)\right] \Omega^\dagger\\
    &=\sum_{g_l,g_r,g_l',g_r'}\omega(g_r,g_r^{-1}g_l)\ket{g_l,g_r}\bra{g_l,g_r}\\
    &\times \left[D(g)\otimes D(g)\right]  \overline{\omega(g_r',g_r'^{-1}g_l')}\ket{g_l',g_r'}\bra{g_l',g_r'}\\
    &=\sum_{g_l',g_r'}\frac{\omega(gg_r',g_r'^{-1}g_l')}{\omega(g_r',g_r'^{-1}g_l')}\ket{g g_l',g g_r'}\bra{g_l',g_r'}\\
    &=\sum_{g_l',g_r'}\frac{\omega(g,g_l')}{\omega(g,g_r')}\ket{g g_l',g g_r'}\bra{g_l',g_r'}\\
    &=\left[\sum_{g_l}\omega(g,g_l)\ket{g g_l}\bra{g_l}\right]\otimes\left[\sum_{g_r} \overline{\omega(g,g_r)}\ket{g g_r}\bra{g_r}\right].
\end{split}
\end{equation}

We move on to consider the effect of TRS. It turns out that, just like the case without symmetry fractionalization, 
the constraint from TRS produces the relation $c_{L,\tilde{g},R}=\overline{c_{L,g,R}}$. This can be seen from the fact that $[D(t)\otimes D(t),\Omega]=0$:
\begin{equation}
\begin{split}
&[D(t)\otimes D(t)]\Omega[D(t)\otimes D(t)]\\
&=\sum_{g_l,g_r\in G_0}\overline{\omega(g_r,g_r^{-1}g_l)}\ket{\tilde g_l,\tilde g_r}\bra{\tilde g_l,\tilde g_r} \\
&=\sum_{\tilde g_l,\tilde g_r\in G_0}\omega(\tilde g_r,\tilde g_r^{-1}\tilde g_l)\ket{\tilde g_l,\tilde g_r}\bra{\tilde g_l,\tilde g_r}=\Omega,
\end{split}
\end{equation}
where we have used $\omega(g,h)\omega(\tilde g,\tilde h)=1$ $\forall g,h\in G_0$. This implies $\forall g=g_0t\in G\setminus G_0$, $\Omega \left[D(g)\otimes D(g)\right] \Omega^\dag=\mathcal{D}(g_0)D(t)\otimes \overline{\mathcal{D}(g_0)}D(t)$.

Without the fractionalization of TRS, suppose $g=g_0\tau\, (\tau=e\,\mathrm{or}\,t)$, $\mathcal{D}(g)$ is identified as
\begin{equation}
\label{eq:projregular}
    \mathcal{D}(g)=\sum_{g'\in G_0}\omega(g_0,g')\ket{g_0 g'}\bra{g'}D(\tau),
\end{equation}
which can be confirmed to be a projective representation with cocycle $\omega$: $\mathcal{D}(g)\mathcal{D}(h)=\omega(g,h)\mathcal{D}(gh)$. 
In fact, in case of the product of $g=g_0$ and $h=h_0\tau'$,
\begin{equation}
\label{eq:projective_regular_proof1}
\begin{split}
    &\mathcal{D}(g)\mathcal{D}(h)\\
    &=\sum_{g',h'}\omega(g_0,g')\omega(h_0,h')\ket{g_0 g'}\langle g'|h_0 h'\rangle
    \bra{h'}D(\tau')\\
    &=\sum_{h'}\omega(g_0,h_0 h' )\omega(h_0,h')\ket{g_0 h_0 h'}\bra{h'}D(\tau')\\
    &=\omega(g_0,h_0)\sum_{h'}\omega(g_0 h_0,h')\ket{g_0 h_0 h'}\bra{h'}D(\tau')\\
    &=\omega(g_0,h_0)\mathcal{D}(gh).
\end{split}
\end{equation}
Otherwise, namely, if $g=g_0t$ and $h=h_0\tau'$,
\begin{equation}
\label{eq:projective_regular_proof2}
\begin{split}
    &\mathcal{D}(g)\mathcal{D}(h)\\
    &=\sum_{g',h'}\frac{\omega(g_0,g')}{\omega(h_0,h')}\ket{g_0 g'}\langle g'|\widetilde{h_0 h'}\rangle
    \bra{h'}D(\tau')\\
    &=\sum_{g',h'}\frac{\omega(g_0,g')}{\omega(h_0,\tilde{h'})}\ket{g_0 g'}\langle g'|\tilde{h_0} h'\rangle
    \bra{\tilde{h'}}D(\tau')\\
    &=\sum_{h'}\omega(g_0,\tilde{h_0} h' )\omega(\tilde{h_0},h')\ket{g_0 \tilde{h_0} h'}\bra{h'}D(t\tau')\\
    &=\omega(g_0,\tilde{h_0})\sum_{h'}\omega(g_0\tilde{h_0},h')\ket{g_0\tilde{h_0} h'}\bra{h'}D(t\tau')\\
    &=\omega(g_0,\tilde{h_0})\mathcal{D}(gh).
\end{split}
\end{equation}

Since we did not fractionalize TRS, $\omega(t,t)$ is 1. By using $\omega(t,t)=1$ and decoupling formulae (\ref{eq:decouplings}), one can show $\omega(g_0,h_0)=\omega(g_0,h_0t)$ and $\omega(g_0t,h_0)=\omega(g_0,\tilde{h_0})=\omega(g_0t,h_0t)$. By summarizing all the above situations, $\mathcal{D}(g)\mathcal{D}(h)=\omega(g,h)\mathcal{D}(gh)$ is proven.

One can further fractionalize the TRS by applying the following $\Upsilon$ gate to 2-dimensional subsystems in $L$ and $R$ of $\Omega\ket{\Psi}$:
\begin{equation}
    \Upsilon=\frac{1-i}{2}\left[\mathbb{1}_4-i(\sigma_y\otimes\sigma_y)\right].
\end{equation}
The fractionalization by $\Upsilon$ enfolds two 2-dimensional subsystems in $L$ and $R$ into the projective regular representation of $G_0\rtimes \mathbb{Z}_2^\mathcal{T}$.
To see this, we consider the action of $\Upsilon$ on $\mathbb{1}_L\otimes \mathcal{D}(g)\otimes \overline{\mathcal{D}(g)}\otimes \mathbb{1}_R$. Since $\Upsilon$ acts only on $L$ and $R$, unitary part $\mathcal{D}(g_0)\otimes\overline{\mathcal{D}(g_0)}$ commutes with $\Upsilon$. On the contrary, $D(t)\otimes D(t)$ includes complex conjugate operation, which is not commutable with $\Upsilon$. Therefore, only the fractionalization of $\mathbb{1}_L\otimes D(t)\otimes D(t)\otimes \mathbb{1}_R$ is essential. The result is 
\begin{equation}
    \begin{split}
        &\Upsilon\left[\mathbb{1}_L\otimes D(t)\otimes D(t)\otimes \mathbb{1}_R\right]\Upsilon^\dag\\
        &=\mathbb{1}_{L'}\otimes i\sigma_y \otimes D(t)\otimes D(t)\otimes i\sigma_y\otimes \mathbb{1}_{R'},
    \end{split}
\end{equation}
which implies the action of $\mathbb{Z}_2^\mathcal{T}$ is $D(t)\otimes i\sigma_y$.

We discuss the effect of $\Upsilon$ on the calculation rule of projective representations in Eqs.~(\ref{eq:projective_regular_proof1}) and (\ref{eq:projective_regular_proof2}). The only change is the additional sign in Eq.~(\ref{eq:projective_regular_proof2}). By considering it in conjunction with Eq.~(\ref{eq:projective_regular_proof1}), we obtain the conclusion that $\mathcal{D}(g)\mathcal{D}(h)=\omega(g,h)\mathcal{D}(gh)=\omega(g_0,\tilde{h_0})\omega(\tau,\tau')\mathcal{D}(gh)$ for $g=g_0\tau$ and $h=h_0\tau'$ where $g_0,h_0\in G_0$ and $\tau,\tau'\in \mathbb{Z}_2^\mathcal{T}$. 
Here $\omega(t,t)=1$ in the absence of $\Upsilon$, otherwise $\omega(t,t)=-1$.
That is why the action of $\Upsilon$ can be interpreted as the symmetry fractionalization of $\mathbb{Z}_2^\mathcal{T}$ part for $G_0\rtimes \mathbb{Z}_2^\mathcal{T}$.

\section{Classification of irreps}
It is known that the Frobenius-Schur indicator determines whether a unitary irrep $\alpha$ is real, complex, or quaternionic \cite{bradley2009mathematical}.
Such an indicator can be generalized to deal with projective irreps, which may further be twisted by an involutory group automorphism induced by TRS  
\cite{kawanaka1990twisted,bradley2009mathematical}:
\begin{equation}
\begin{split}
    &\frac{1}{|G_0|}\sum_{g\in G_0}\omega(tg,tg)\chi_\alpha((tg)^2)\\
    &=\frac{\omega(t,t)}{|G_0|}\sum_{g\in G_0}\omega(\tilde{g},g)\chi_\alpha(\tilde{g}g).
\end{split}
\label{ind}
\end{equation}
Here the sixth decoupling formula in Eq.~(\ref{eq:decouplings}) was used. Focusing on the unitary part, we define the \emph{indicator} of irrep $\alpha$ as follows:
\begin{equation}
    \iota_\alpha=\frac{1}{|G_0|}\sum_{g\in G_0}\omega(\tilde{g},g)\chi_\alpha(\tilde{g}g).
\end{equation}
There are three possibilities: 
\begin{enumerate}
    \item $\iota_\alpha=1\Leftrightarrow \overline{\mathcal{D}^\alpha(\tilde{g})}=\mathcal{D}^\alpha(g)$,
    \item $\iota_\alpha=-1\Leftrightarrow \overline{\mathcal{D}^\alpha(\tilde{g})}=Y\mathcal{D}^\alpha(g)Y,\;\;Y=\sigma_y\otimes\mathbb{1}_{d_\alpha/2}$,
    \item $\iota_\alpha=0\Leftrightarrow$ 
    $\overline{\mathcal{D}^\alpha(\tilde{g})}$ and $\mathcal{D}^\alpha(g)$ are different irreps.
\end{enumerate}
To show these results, we first note that
\begin{equation}
\begin{split}
    \overline{\mathcal{D}^\alpha(\tilde{g})}\overline{\mathcal{D}^\alpha(\tilde{g'})}
    &=\overline{\omega(\tilde{g},\tilde{g'})\mathcal{D}^\alpha(\widetilde{gg'})}\\
    &=\omega(g,g')\overline{\mathcal{D}^\alpha(\widetilde{gg'})},
\end{split}
\end{equation}
where $\omega(g,h)\omega(\tilde g,\tilde h)=1$ has been used. Hence, $\mathcal{D}^{\alpha^\star}
(g):=\overline{\mathcal{D}^\alpha(\tilde{g})}$ should be an irrep with the same cocycle as well. 
Given irrep $\alpha$, consider the following operator:
\begin{equation}
    X^\alpha=\frac{1}{|G_0|}\sum_{g\in G_0}\omega(\tilde{g},g)\mathcal{D}^\alpha (\tilde{g}g),   
\end{equation}
whose trace gives the indicator (\ref{ind}). 
This operator can be rewritten as 
\begin{equation}
    \begin{split}
    X^\alpha
    &=\frac{1}{|G_0|}\sum_{g\in G_0}\omega(\tilde{g},g)\mathcal{D}^\alpha (\tilde{g}g)\\
    &=\frac{1}{|G_0|}\sum_{g\in G_0}\mathcal{D}^\alpha (\tilde{g})\mathcal{D}^\alpha(g)\\
    &=\frac{1}{|G_0|}\sum_{g\in G_0}\overline{\mathcal{D}^{\alpha^\star} (g)}\mathcal{D}^\alpha(g).
    \end{split}
\end{equation}

If $\alpha^\star\nsim \alpha$, $X^\alpha=0$ follows from Theorem~\ref{thm:orthogonality}, implying $\iota=0$. 
Conversely, if $\iota_\alpha=0$ we necessarily have $\alpha^\star\nsim \alpha$, since otherwise $\iota_\alpha=\pm1$, as will be clear in the following. 
If $\alpha^\star\sim \alpha$, the irreps are related by a unitary transformation:
$\mathcal{D}^{\alpha^\star}(g)=\overline{\mathcal{D}^{\alpha}(\tilde g)}=u^\dag\mathcal{D}^\alpha(g) u.$
By using this reation again, we get
\begin{equation}
    \mathcal{D}^{\alpha}(g)=u^{\rm T}\overline{\mathcal{D}^\alpha(\tilde g)}\overline{u}=u^{\rm T} u^\dag\mathcal{D}^\alpha(g)u \overline{u},
\end{equation}
which implies $u\overline{u}=(uK)^2=\lambda \mathbb{1}_{d_\alpha}$ with $|\lambda|=1$ and $K$ being the complex conjugation.
Considering $(uK)^3=(uK)^2uK=uK(uK)^2$, we get $\lambda=\pm 1$. In this case, using the fact that $\{\mathcal{D}^\alpha(g)\}_{g\in G_0}$ is a unitary $1$-design, so that 
\begin{equation}
\frac{1}{|G_0|}\sum_{g\in G_0}\overline{\mathcal{D}^{\alpha} (g)}A \mathcal{D}^\alpha(g)=\frac{1}{d_\alpha} A^{\rm T}
\end{equation}
for any operator $A$ on the vector space of irrep $\alpha$, as is clear from the graphical representation  
\begin{equation}
\begin{tikzpicture}
\draw[ultra thick] (0.3,-1) -- (0.3,1) (1,-1) -- (1,1) (1.7,-1) -- (1.7,1);
\draw[ultra thick] (0.3,1) .. controls (0.3,1.3) and (1,1.3) .. (1,1);
\draw[ultra thick] (1,-1) .. controls (1,-1.3) and (1.7,-1.3) .. (1.7,-1);
\draw[thick,fill=blue!10!white] (0,0) rectangle (0.6,0.6);
\draw[thick,fill=blue!10!white] (1.4,0) rectangle (2,0.6);
\draw[thick,fill=red!10!white] (1.7,-0.5) circle (0.3);
\Text[x=-0.5,y=0,fontsize=\large]{$\mathbb{E}_{U}$}
\draw[thick] (-0.1,-1.2) -- (-0.2,-1.2) -- (-0.2,1.2) -- (-0.1,1.2);
\draw[thick] (2.1,-1.2) -- (2.2,-1.2) -- (2.2,1.2) -- (2.1,1.2);
\Text[x=0.3,y=0.3,fontsize=\large]{$U$}
\Text[x=1.7,y=0.335,fontsize=\large]{$\overline{U}$}
\Text[x=1.7,y=-0.5,fontsize=\large]{$A$}
\end{tikzpicture}
\;\;\;\;
\begin{tikzpicture}
\draw[ultra thick] (0.3,0) .. controls (0.3,0.3) and (1.7,0.3) .. (1.7,0);
\draw[ultra thick] (0.3,0.6) .. controls (0.3,0.3) and (1.7,0.3) .. (1.7,0.6);
\fill[white] (0.9,0) rectangle (1.1,0.5);
\draw[ultra thick] (0.3,-1) -- (0.3,0) (0.3,0.6) -- (0.3,1) (1,-1) -- (1,1) (1.7,-1) -- (1.7,0) (1.7,0.6) -- (1.7,1);
\draw[ultra thick] (0.3,1) .. controls (0.3,1.3) and (1,1.3) .. (1,1);
\draw[ultra thick] (1,-1) .. controls (1,-1.3) and (1.7,-1.3) .. (1.7,-1);
\draw[thick,fill=red!10!white] (1.7,-0.5) circle (0.3);
\Text[x=-0.4,y=0,fontsize=\large]{$=\frac{1}{d_\alpha}$}
\Text[x=1.7,y=-0.5,fontsize=\large]{$A$}
\end{tikzpicture}.
\end{equation}
Accordingly, we have
\begin{equation}
\begin{split}
X^\alpha
    &=\frac{1}{|G_0|}\sum_{g\in G_0}u^\dag\overline{\mathcal{D}^{\alpha} (g)}u \mathcal{D}^\alpha(g)\\
    &=\frac{1}{d_\alpha}(u\overline{u})^{\rm T}
     =\lambda\frac{\mathbb{1}_{d_\alpha}}{d_\alpha},
\end{split}    
\end{equation}
implying $\iota_\alpha=\lambda$. Since we can always choose an appropriate basis such that $u=\mathbb{1}_{d_\alpha}$ ($u=\sigma_y\otimes\mathbb{1}_{d_\alpha/2}$) for $\lambda=1$ ($\lambda=-1$), we obtain the desired results.

\section{Detailed proof of theorem \ref{thm:espectrum}} \label{sec:mainproof}

The reduced density matrix on $L\cup l$ of $\Omega\ket{\Psi}$ is $\rho = W W^\dagger$, where the matrix elements $[W]_{L,g_l|g_r,R}=\langle L|\langle g_l|W|g_r\rangle|R\rangle$ are given by
\begin{equation}
    [W]_{L,g_l|g_r,R} = \frac{c_{L,g_r^{-1}g_l,R} \omega(g_r, g_r^{-1}g_l)}{\sqrt{|G_0|}}.
\end{equation}

\subsection{Partial diagonalization}

We show that $W_{L,R}$, which is a $|G_0|\times|G_0|$ block with entries $[W_{L,R}]_{g_l,g_r}=[W]_{L,g_l|g_r,R}$, is partially diagonalizable with the following matrix:
\begin{equation}
    [U]_{\theta,g}=\sqrt{\frac{d_\alpha}{|G_0|}}\mathcal{D}^\alpha_{ij}(g),
\end{equation}
where $\alpha$ labels an irrep, $d_\alpha$ is the dimension of irrep $\alpha$, $i,j$ are matrix indices of irrep $\alpha$, and $\theta=(\alpha,i,j)$. The square sum theorem $\sum_\alpha {d_\alpha}^2=|G_0|$ guarantees $U$ is a square matrix.  Moreover, one can check using Theorem~\ref{thm:orthogonality} that $U$ is unitary: 
\begin{equation}
\begin{split}
    [UU^\dagger]_{\theta,\theta'}&=\frac{\sqrt{d_\alpha d_{\alpha'}}}{|G_0|}\sum_g \mathcal{D}^\alpha_{i j}(g)\overline{\mathcal{D}^{\alpha'}_{i' j'}(g)}\\
    &=\frac{\sqrt{d_\alpha d_{\alpha'}}}{|G_0|}\frac{|G_0|}{d_\alpha}\delta_{\alpha,\alpha'}\delta_{i,i'}\delta_{j,j'}\\
    &=\delta_{\theta,\theta'}.
\end{split}
\end{equation}

\begin{figure}
\begin{center}
       \includegraphics[width=7cm, clip]{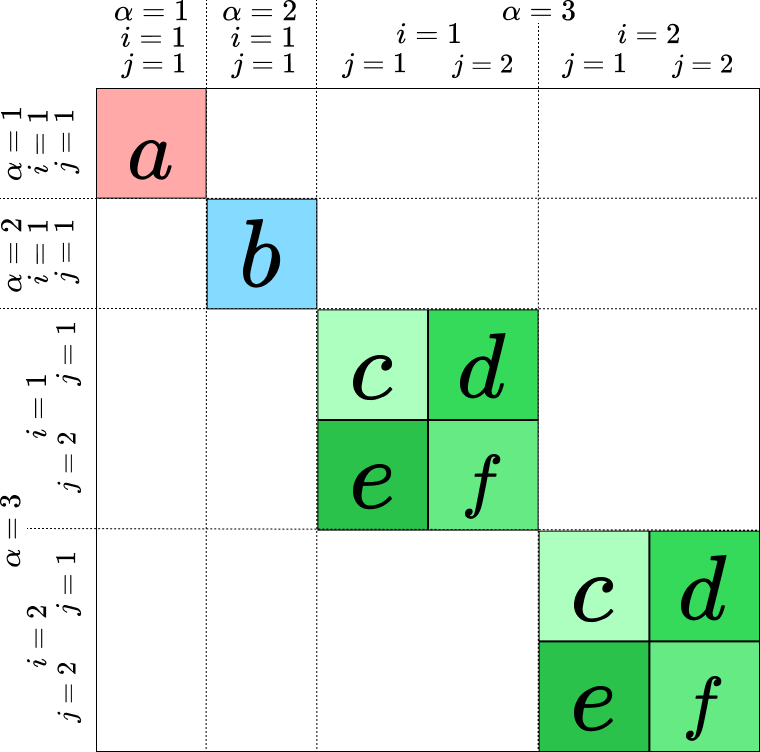}
       \end{center}
   \caption{Configuration of matrix block $W_{L,R}$. Here $G=G_0=C_{3v}$ was chosen for an example. $C_{3v}$ has three irreducible linear representations, two ($\alpha=1,2$) of which are 1D and the other $\alpha=3$ is 2D. $\alpha=1,2$ are 1D, thus they correspond to two independent $1\times1$ blocks $a$ and $b$. $\alpha=3$ is 2D, thus it corresponds to two identical $2\times2$ blocks $\begin{bmatrix} c & d \\ e & f \end{bmatrix}$.} 
   \label{fig:WLR}
\end{figure}

Applying $U$ conjugation to $W_{L,R}$, we get
\begin{equation}
\begin{split}
    &[U W_{L,R} U^\dagger]_{\theta,\theta'}\\
    &=\frac{\sqrt{d_\alpha d_{\alpha'}}}{|G_0|^{3/2}}\sum_{g_l,g_r}\mathcal{D}^\alpha_{i j}(g_l)c_{L,g_r^{-1}g_r,R}\omega(g_r,g_r^{-1}g_l)\overline{\mathcal{D}^{\alpha'}_{i'j'}(g_r)}.
\end{split}
\end{equation}
Retaking $g=g_r^{-1}g_l$ and using $\mathcal{D}(g)\mathcal{D}(h)=\omega(g,h)\mathcal{D}(gh)$, one can get
\begin{equation}
\label{eq:partial_diag}
\begin{split}
    &[U W_{L,R} U^\dagger]_{\theta,\theta'}\\
    &=\frac{\sqrt{d_\alpha d_{\alpha'}}}{|G_0|^{3/2}}\sum_{g,g_l}\mathcal{D}^\alpha_{i j}(g_l)c_{L,g,R}\omega(g_lg^{-1},g)\overline{\mathcal{D}^{\alpha'}_{i'j'}(g_lg^{-1})}\\
    &=\frac{\sqrt{d_\alpha d_{\alpha'}}}{|G_0|^{3/2}}\sum_{g,g_l,j''}\mathcal{D}^\alpha_{i j}(g_l)c_{L,g,R}\frac{\overline{\mathcal{D}^{\alpha'}_{i'j''}(g_l)\mathcal{D}^{\alpha'}_{j''j'}(g^{-1})}}{\overline{\omega(g^{-1},g)}}\\
    &=\frac{1}{\sqrt{|G_0|}}\sum_{g,j''}c_{L,g,R}\delta_{\alpha,\alpha'}\delta_{i, i'}\delta_{j, j''}\omega(g^{-1},g)\overline{\mathcal{D}^{\alpha'}_{j''j'}(g^{-1})}\\
    &=\frac{\delta_{\alpha,\alpha'}\delta_{i, i'}}{\sqrt{|G_0|}}\sum_{g}c_{L,g,R}\omega(g^{-1},g)\overline{\mathcal{D}^{\alpha'}_{jj'}(g^{-1})}\\
    &=\frac{\delta_{\alpha,\alpha'}\delta_{i,i'}}{\sqrt{|G_0|}}\sum_g c_{L,g,R}\mathcal{D}^\alpha_{j'j}(g).
\end{split}
\end{equation}
The final form involves $\delta_{\alpha,\alpha'}\delta_{i,i'}$,  
so $W_{L,R}$ is partially diagonalized into a direct sum of blocks labeled by $\alpha$ and $i$. In each block, a matrix element is specified by the remaining indices $j$ and $j'$ (on top of $L,R$, if we return to the entire $W$). See Fig.~\ref{fig:WLR} for an example. As the expression turns out to be $i$-independent, all the blocks with the same $\alpha$ but different $i$ should be exactly the same. Since the degeneracy is exactly $d_\alpha$, we have the following direct-sum decomposition: 
\begin{equation}
W=\bigoplus_\alpha \mathbb{1}_{d_\alpha}\otimes W_\alpha,
\label{Wdec}
\end{equation}
where $W_\alpha$ is a $d_Ld_\alpha\times d_Rd_\alpha$ matrix (see the lower panel in Fig.~\ref{fig:rearrangeW}) with entries
\begin{equation}
W_{L,R,\alpha,j,j'}=\frac{1}{\sqrt{|G_0|}}\sum_g c_{L,g,R}\mathcal{D}^\alpha_{j'j}(g).
\label{WLRajj}
\end{equation}

Accordingly, $\rho=W W^\dagger$ is decomposed into
\begin{equation}
    \rho=\bigoplus_\alpha \mathbb{1}_{d_\alpha}\otimes \rho_\alpha, 
\end{equation}
where $\rho_\alpha$ is a $d_L d_\alpha\times d_L d_\alpha$ matrix.

\subsection{Identification of the statistics}

This section identifies what ensemble each component of the direct sum decomposition (\ref{Wdec}) and the corresponding entanglement spectrum obey.
Since $c_{L,g,R}$ are sampled from identical independently distributed (i.i.d.) complex Gaussian variables $\mathcal{CN}(0,1)$, and all the matrix elements are their linear combinations, we can 
consider the covariance matrix to completely determine the statistics. In particular, a pair of complex Gaussian random variables
$a,b$ are independent iff
$\mathbb{E}(a\overline{b})=\mathbb{E}(ab)=0$.
In addition, to confirm the independence of the real and imaginary parts of each variable (and that their variances are equal), we further require $\mathbb{E}(a^2)=0$. 

\subsubsection{$G=G_0$}

First, we consider a general unitary symmetry $G=G_0$. Since $c_{L,g,R}$ are i.i.d., matrix elements with different $L$ or $R$ are obviously independent. Given $L$ and $R$, the covariance matrix 
reads
\begin{equation}
\begin{split}
    &\mathbb{E}\left[W_{L,R,\alpha,j,j'}\overline{W_{L,R,\alpha',l,l'}}\right]\\
    &=\frac{1}{|G_0|}\mathbb{E}\left[\left(\sum_g c_{L,g,R}\mathcal{D}^\alpha_{j' j}(g)\right)\overline{\left(\sum_{g'} c_{L,g',R}\mathcal{D}^{\alpha'}_{l' l}(g')\right)}\right]\\
    &=\frac{1}{|G_0|}\sum_{g,g'} \mathbb{E}\left[c_{L,g,R}\overline{c_{L,g',R}}\right]
    \mathcal{D}^\alpha_{j' j}(g)\overline{\mathcal{D}^{\alpha'}_{l' l}(g')}\\
    &=\frac{1}{d_\alpha}\delta_{\alpha,\alpha'}\delta_{j l}\delta_{j' l'},
\end{split}
\label{EWbW}
\end{equation}
where we have used $\mathbb{E}[c_{L,g,R}\overline{c_{L,g',R}}]=\delta_{g,g'}$ and Theorem~\ref{thm:orthogonality}. On the other hand, we have 
\begin{equation}
\begin{split}
    &\mathbb{E}\left[W_{L,R,\alpha,j,j'}W_{L,R,\alpha',l,l'}\right]\\
    &=\frac{1}{|G_0|}\mathbb{E}\left[\left(\sum_g c_{L,g,R}\mathcal{D}^\alpha_{j' j}(g)\right)\left(\sum_{g'} c_{L,g',R}\mathcal{D}^{\alpha'}_{l' l}(g')\right)\right]\\
    &=\frac{1}{|G_0|}\sum_{g,g'}\mathbb{E}\left[c_{L,g,R}c_{L,g',R}\right]\mathcal{D}^\alpha_{j' j}(g)\mathcal{D}^{\alpha'}_{l' l}(g')\\
    &=0,
\end{split}
\label{EWW}
\end{equation}
which follows from $\mathbb{E}\left[c_{L,g,R}c_{L,g',R}\right]=0$. Therefore, there is no correlation between different matrix elements. 
The covariance matrix indicates the matrix element $W_{L,R,\alpha,j,j'}$ obey i.i.d. $\mathcal{N}(0,1/d_\alpha)$.

We recall that exactly the same square complex Gaussian random matrix blocks are repeated $d_\alpha$ times for an index $\alpha$ (cf. Fig.~\ref{fig:WLR}). This implies that, if $G=G_0$, the matrix ensemble of $WW^\dag$ (\ref{Wdec}) 
is identified as
\begin{equation}
\label{eq:unitary_decomposition}
    \bigoplus_{\alpha}
    \left[
        \frac{\mathbb{1}_{d_\alpha}}{d_\alpha} \otimes \mathrm{LUE}^{d_Ld_\alpha\times d_Rd_\alpha}_\alpha
    \right],
\end{equation}
whose singular value distribution conditioned on the normalization constraint gives the statistics of the entanglement spectrum.

Finally, we take a glance at the joint distribution $\rho(\lambda_1,\dots)$ of the entanglement spectrum.
By assuming ($d_L\leq d_R$), the nonzero eigenvalues $\{\lambda_i\}_{i=1}^{d_L|G_0|}$ of the reduced density matrix satisfy $\sum_{i=1}^{d_L|G_0|} \lambda_i=1$. Therefore the joint distribution contains $\delta\left(\sum_{i=1}^{d_L|G_0|} \lambda_i-1\right)$. Moreover, according to Eq.~(\ref{eq:unitary_decomposition}), the eigenvalues can be relabeled as follows:
\begin{gather}
    \lambda_1,\dots,\lambda_{d_L|G_0|}\nonumber\\
    \downarrow\\
    \lambda_{1,\alpha_1},\dots, \lambda_{d_Ld_{\alpha_1},\alpha_1},\lambda_{1,\alpha_2},\dots, \lambda_{d_Ld_{\alpha_2},\alpha_2},\dots\nonumber
\end{gather}
Note that the eigenvalue $\lambda_{i,\alpha}$ 
is $d_\alpha$-fold degenerate, thus the normalization condition is rewritten as
\begin{equation}
    \sum_\alpha d_\alpha \sum_{i=1}^{d_Ld_\alpha} \lambda_{i,\alpha}=1.
    \label{da1}
\end{equation}

The standard deviation of $W_{L,R,\alpha,j,j'}$ is $d_\alpha^{-1}$, 
thus the exponential weight is $e^{-\sum_\alpha d_\alpha \sum_{i=1}^{d_\alpha d_L}\lambda_{i,\alpha}}$, which can be absorbed into the normalization factor due to Eq.~(\ref{da1}).
Taking everything in consideration, we get the joint distribution of entanglement spectrum of $W W^\dagger$:
\begin{equation}
\begin{split}
    &p(\{\lambda_{i,\alpha}\})\propto
    \delta\left(
        \sum_\alpha d_\alpha \sum_{i=1}^{d_L d_\alpha}\lambda_{i,\alpha}-1
    \right)\\
    &\times 
    \prod_\alpha\left[
        \prod_{i=1}^{d_L d_\alpha} \lambda_{i,\alpha}^{d_\alpha (d_R-d_L)}
        \prod_{i<j}|\lambda_{i,\alpha}-\lambda_{j,\alpha}|^2
    \right].
\end{split}
\end{equation}

\begin{figure}
\begin{center}
       \includegraphics[width=8cm, clip]{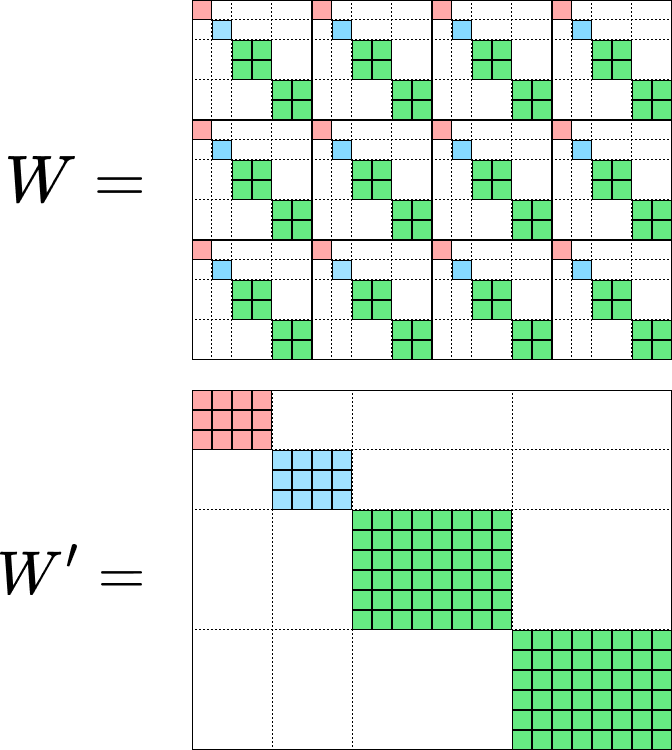}
       \end{center}
   \caption{Let $G=G_0=C_{3v}$ again, and let $d_L=3,d_R=4$. Rearranging $W\mapsto W'$, we get block-diagonalized $d_L|G_0|\times d_R|G_0|$ rectangular matrix. The first (red) and the second (blue) blocks are mutually independent, while the third and the fourth (green) blocks are identical.}
   \label{fig:rearrangeW}
\end{figure}

\subsubsection{$G=G_0\rtimes\mathbb{Z}_2^\mathcal{T}$ in the absence of $\Upsilon$}

Next we consider $G=G_0\rtimes \mathbb{Z}_2^\mathcal{T}$ (or $G=G_0\times \mathbb{Z}_2^\mathcal{T}$). In the absence of $\Upsilon$, partial diagonalization is already achieved by Eq.~(\ref{eq:partial_diag}). 
To determine the ensemble, it is necessary to explore the effect of time reversal on different irrep blocks. 
Taking the complex conjugation of Eq.~(\ref{WLRajj}), we obtain
\begin{equation}
\label{eq:partial_diag_cc}
\begin{split}
    \overline{\sum_g c_{L,g,R}\mathcal{D}^\alpha(g)}&=\sum_g c_{L,\tilde{g},R}\overline{\mathcal{D}^\alpha(g)}\\
    &=\sum_g c_{L,g,R}\overline{\mathcal{D}^\alpha(\tilde{g})}\\
    &=\sum_g c_{L,g,R}\mathcal{D}^{\alpha^\star}
    (g),
\end{split}
\end{equation}
where we have used $c_{L,\tilde g,R}=\overline{c_{L,g,R}}$. The three possibilities between $\mathcal{D}^\alpha$ and $\mathcal{D}^{\alpha^\star}$ can then be translated to $W_\alpha$. It is natural to expect that $W_\alpha$ is a LOE (LSE) block if $\iota_\alpha=1$ ($\iota_\alpha=-1$), and otherwise a LUE block conjugate to another one $W_{\alpha^\star}$ if $\iota_\alpha=0$. 

To verify the above conjecture, we have to calculate the covariance matrix. The only differences from the case of unitary symmetry are the constraints $c_{L,\tilde{g},R}=\overline{c_{L,g,R}}$ and $\omega(g,h)\omega(\tilde{g},\tilde{h})=1$. 
The first-type covariance matrix given in Eq.~(\ref{EWbW}) stays exactly the same
under these constraints. However, the second-type covariant matrix (\ref{EWW}) now becomes nontrivial in general: 
\begin{equation}
\label{eq:timerev_corr}
\begin{split}
    &\mathbb{E}\left[W_{L,R,\alpha,j,j'}W_{L,R,\alpha',l,l'}\right]\\
    &=\frac{1}{|G_0|}\mathbb{E}\left[\left(\sum_g c_{L,g,R}\mathcal{D}^\alpha_{j' j}(g)\right)\left(\sum_{g'} c_{L,g',R}\mathcal{D}^{\alpha'}_{l' l}(g')\right)\right]\\
    &=\frac{1}{|G_0|}\sum_{g,g'} \mathbb{E}\left[c_{L,g,R}\overline{c_{L,\tilde{g'},R}}\right] \mathcal{D}^\alpha_{j' j}(g)\mathcal{D}^{\alpha'}_{l' l}(g')\\
    &=\frac{1}{|G_0|}\sum_{g,g'} \delta_{g,\tilde{g'}} \mathcal{D}^\alpha_{j' j}(g)\mathcal{D}^{\alpha'}_{l' l}(g')\\
    &=\frac{1}{|G_0|}\sum_{g'} \mathcal{D}^\alpha_{j' j}(\tilde{g'})\mathcal{D}^{\alpha'}_{l' l}(g').
\end{split}
\end{equation}
For $\iota_\alpha=1$, $\mathcal{D}^\alpha(\tilde{g})=\overline{\mathcal{D}^\alpha(g)}$ for some proper basis, thus Eq.~(\ref{eq:partial_diag_cc}) implies the matrix elements are real. The covariance matrix is
\begin{equation}
\begin{split}
    &\mathbb{E}\left[W_{L,R,\alpha,j,j'}W_{L,R,\alpha',l,l'}\right]\\
    &=\frac{1}{|G_0|}\sum_{g'} \overline{\mathcal{D}^{\alpha}_{j' j}(g')}\mathcal{D}^{\alpha'}_{l' l}(g')\\
    &=\frac{1}{d_{\alpha'}}\delta_{\alpha,\alpha'}\delta_{j,l}\delta_{j',l'},
\end{split}
\end{equation}
which is the same as $\mathbb{E}\left[W_{L,R,\alpha,j,j'}\overline{W_{L,R,\alpha',l,l'}}\right]$, consistent with the realness of matrix elements. 
This implies no correlation between different matrix elements, so the real block $\alpha$ obeys the LOE.

For $\iota_\alpha=-1$, $\mathcal{D}^\alpha(\tilde{g})=Y\overline{\mathcal{D}^\alpha(g)}Y$ for some proper basis, thus Eq.~(\ref{eq:partial_diag_cc}) implies $2\times 2$ blocks in $W_{L,R}$ are spin representations of quaternions. The covariance matrix is 
\begin{equation}
\begin{split}
    &\mathbb{E}\left[W_{L,R,\alpha,j,j'}W_{L,R,\alpha',l,l'}\right]\\
    &=\frac{1}{|G_0|}\sum_{a,b,g'} Y_{j',a}\overline{\mathcal{D}^{\alpha}_{a,b}(g')}Y_{b,j}\mathcal{D}^{\alpha'}_{l' l}(g')\\
    &=\frac{1}{d_{\alpha'}}\sum_{a,b}Y_{j',a}Y_{b,j}\delta_{\alpha,\alpha'}\delta_{b,l}\delta_{a,l'}\\
    &=\frac{1}{d_{\alpha'}}\delta_{\alpha,\alpha'}Y_{l,j}Y_{j',l'}.
\end{split}
\end{equation}
This result indicates no correlation exists between different $2\times 2$ blocks in Eq.~(\ref{WLRajj}), simply due to that $Y_{j,l}$ is zero if $j$ and $l$ are not in the same block. Moreover, the factor $Y_{l,j}Y_{j',l'}$ is consistent with the spin representations of quaternion, so such a block obeys the LSE.

For $\iota_\alpha=0$, $\overline{\mathcal{D}^\alpha(\tilde{g})}$ is a different irrep $\alpha^\star$ from $\alpha$, thus
\begin{equation}
\begin{split}
    &\mathbb{E}\left[W_{L,R,\alpha,j,j'}W_{L,R,\alpha',l,l'}\right]\\
    &=\frac{1}{|G_0|}\sum_{g'} \overline{\mathcal{D}^{\alpha^\star}_{j' j}(g')}\mathcal{D}^{\alpha'}_{l' l}(g')\\
    &=\frac{1}{d_{\alpha'}}\delta_{\alpha^\star,\alpha'}\delta_{j,l}\delta_{j',l'}.
\end{split}
\end{equation}
This result indicates that every matrix element is complex and independent, so such a block obeys the LUE. Moreover, after an appropriate unitary transformation, blocks $\alpha$ and $\alpha^\star$ can always be made a complex conjugation pair.

Taking into account all the previous discussions, we can identify the matrix ensemble of $WW^\dag$ (\ref{Wdec}) as follows: 
\begin{gather}
    \left[\bigoplus_{\alpha:R_1}\frac{\mathbb{1}_{d_\alpha}}{d_\alpha} \otimes \mathrm{LOE}_\alpha^{d_Ld_\alpha\times d_Rd_\alpha}\right]\nonumber\\
    \oplus \nonumber\\
    \left[\bigoplus_{\alpha:R_0}\frac{\mathbb{1}_{d_\alpha}}{d_\alpha} \otimes \left(\mathrm{LUE}_\alpha^{d_Ld_\alpha\times d_Rd_\alpha}\oplus\overline{\mathrm{LUE}_\alpha^{d_Ld_\alpha\times d_Rd_\alpha}}\right)\right]\nonumber\\
    \oplus\nonumber\\
    \left[\bigoplus_{\alpha:R_{-1}}\frac{\mathbb{1}_{d_\alpha}}{d_\alpha} \otimes \mathrm{LSE}_\alpha^{d_Ld_\alpha\times d_Rd_\alpha}\right].\label{eq:without_up}
\end{gather}

We discuss the joint distribution of entanglement spectrum. From Eq.~(\ref{EWbW}), the variance of matrix elements $W_{L,R,\alpha,j,j'}$ is $d_\alpha^{-1}$ for $\alpha\in R_+$ or $\alpha\in R_0$. However, for $\alpha\in R_-$, $2\times 2$ matrix elements represent one quaternion. Here, quaternion $q$ is represented as $q\to\begin{pmatrix}a&b\\-\overline{b}&\overline{a}\end{pmatrix}$, where $a,b$ are independent random Gaussian complex numbers with variance $1/d_\alpha$. This is equivalent to $q=\mathfrak{Re}a\cdot 1+\mathfrak{Im}a\cdot\mathbb{i}+\mathfrak{Re}b\cdot \mathbb{j}+\mathfrak{Im}b\cdot\mathbb{k}$, where $\mathbb{i,j,k}$ are the basis vectors of quaternion. Thus, one can realize $\mathbb{E}[|q|^2]=\mathbb{E}[|a|^2+|b|^2]=2/d_\alpha$, which is different from that of $R_+$ or $R_0$.
To obtain the joint distribution, we review the difference of Gaussian probability density function (p.d.f.) between real, complex and quaternionic numbers.
One can check that the following p.d.f. produces random Gaussian variables with mean $0$ and variance $\sigma^2, \sigma^2, 2\sigma^2$ for real, complex, quaternionic numbers, respectively:
\begin{equation}
\begin{split}
    &\frac{1}{\sqrt{2\pi}\sigma}e^{-\frac{x^2}{2\sigma^2}}\quad \mathrm{for \; real\; number\;} x,\\
    &\frac{1}{\pi\sigma^2}e^{-\frac{|z|^2}{\sigma^2}}\quad \mathrm{for \; complex\; number\;} z,\\
    &\frac{1}{\pi^2\sigma^4}e^{-\frac{|q|^2}{\sigma^2}}\quad \mathrm{for \; quaternionic\; number\;} q.
\end{split}
\end{equation}
Under the circumstances of this case, $\sigma^2$ is replaced by $1/d_\alpha$. Based on the above discussion, we obtain the joint distribution of entanglement spectrum of $W W^\dagger$ as follows:
\begin{equation}
\label{eq:distrubution_antiunitary}
\begin{split}
    &p(\{\lambda_{i,\alpha}\})\\
    &\propto\delta\left(\sum_{\alpha:R_+} d_\alpha \sum_{i=1}^{d_L d_\alpha}\lambda_{i,\alpha}
        +2\sum_{\alpha:R_0} d_\alpha \sum_{i=1}^{d_L d_\alpha}\lambda_{i,\alpha}\right.\\
        &\quad\quad\quad \left.+2\sum_{\alpha:R_-} d_\alpha \sum_{i=1}^{d_L d_\alpha/2}\lambda_{i,\alpha}-1\right)\\
    &\times 
    \prod_{\alpha:R_+}\left[
        \prod_{i=1}^{d_L d_\alpha} \lambda_{i,\alpha}^{(d_\alpha (d_R - d_L)-1)/2}
        \prod_{i<j}^{d_L d_\alpha}|\lambda_{i,\alpha}-\lambda_{j,\alpha}|
    \right]\\
    &\times 
    \prod_{\alpha:R_0}\left[
        \prod_{i=1}^{d_L d_\alpha} \lambda_{i,\alpha}^{d_\alpha (d_R - d_L)}
        \prod_{i<j}^{d_L d_\alpha}|\lambda_{i,\alpha}-\lambda_{j,\alpha}|^2
    \right]\\
    &\times 
    \prod_{\alpha:R_-}\left[
        \prod_{i=1}^{d_L d_\alpha/2} \lambda_{i,\alpha}^{d_\alpha (d_R - d_L)+1}
        \prod_{i<j}^{d_L d_\alpha/2}|\lambda_{i,\alpha}-\lambda_{j,\alpha}|^4
    \right].
\end{split}
\end{equation}

Note that the LUE blocks always have involution pairs $\{\alpha,\alpha^\star\}$, and the LSE blocks have two-fold degenerate eigevalues which 
originate from the spin representation of quaternions.

\subsubsection{$G=G_0\rtimes\mathbb{Z}_2^\mathcal{T}$ in the presence of $\Upsilon$}

The gate $\Upsilon$ fractionalizes the TRS, and appears in Ref.
\cite{cirac2017matrix} as a building block of nontrivial time-reversal symmetric matrix-product unitary. In the presence of $\Upsilon$, $c_{L,g,R}$ in Eq.~(\ref{WLRajj})
is amended to
\begin{equation}
\begin{split}
    \frac{1-i}{2}\left(c_{L',\sigma_L,g,\sigma_R,R'}+i\sigma_L\sigma_R c_{L',-\sigma_L,g,-\sigma_R,R'}\right),
\end{split}
\end{equation}
where $L=L'\sigma_L$, $R=\sigma_R R'$ with $\sigma_{L,R}=\pm$ spanning a $2$-dimensional subsystem. 
Note that the action of $\Upsilon$ is compatible with 
the partial diagonalization by $U$, as they act on different subsystems. This may be understood from the fact that $\omega$ can be decoupled into the unitary and anti-unitary parts. Now a $2\times 2$ block in the partially diagonalized $W$ consisting of different indices $\sigma_L,\sigma_R$ appears as
\begin{equation}
    \label{eq:lue_with_ups}
    Q=
    \frac{1-i}{2}\begin{pmatrix}
        a+ib&c-id\\
        d-ic&b+ia
    \end{pmatrix}.
\end{equation}
Since $a,b,c,$ and $d$ are independent real or complex random Gaussian variables and have the same variance $1/d_\alpha$ for irrep $\alpha$, one can check $\mathbb{E}[|Q_{ij}|^2]=1/d_\alpha$ as well.
There are three possibilities: this $2\times2$ block is weaved into a LOE block, a LUE block, or a LSE block.

If this $2\times2$ block appears in the LOE parts of Eq.~(\ref{eq:without_up}), the coefficients $a,b,c,d$ are all real. In this case, this $2\times2$ block is the spin representation of a quaternion, which satisfies 
\begin{equation}
    \begin{split}
        \sigma_y \overline{Q} \sigma_y\\
        &=
        \begin{pmatrix}
            0&-1\\
            1&0
        \end{pmatrix}
        \frac{1+i}{2}\begin{pmatrix}
            a-ib&c+id\\
            d+ic&b-ia
        \end{pmatrix}
        \begin{pmatrix}
            0&1\\
            -1&0
        \end{pmatrix}\\
        &=
        \frac{1+i}{2}\begin{pmatrix}
            b-ia&-d-ic\\
            -c-id&a-ib
        \end{pmatrix}\\
        &=
        \frac{1-i}{2}\begin{pmatrix}
            a+ib&c-id\\
            d-ic&b+ia
        \end{pmatrix}\\
        &=Q.
    \end{split}
\end{equation}
Therefore, LOE blocks in Eq.~(\ref{eq:without_up}) are amended to LSE blocks under the action of $\Upsilon$. 

If this $2\times2$ block (Eq.~(\ref{eq:lue_with_ups})) appears in the LUE parts of Eq.~(\ref{eq:without_up}), $a,b,c,d$ will be complex and there will always be another block
\begin{equation}
    Q^\star
    =\frac{1-i}{2}
    \begin{pmatrix}
        \overline{a}+i\overline{b}&\overline{c}-i\overline{d}\\
        \overline{d}-i\overline{c}&\overline{b}+i\overline{a}
    \end{pmatrix},
\end{equation}
as LUE blocks always appear in complex pairs. In this case, the following relation holds:
\begin{equation}
    \begin{split}
        &\sigma_y \overline{Q}\sigma_y\\
        &=
        \begin{pmatrix}
            0&-1\\
            1&0
        \end{pmatrix}
        \frac{1+i}{2}\begin{pmatrix}
            \overline{a}-i\overline{b}&\overline{c}+i\overline{d}\\
            \overline{d}+i\overline{c}&\overline{b}-i\overline{a}
        \end{pmatrix}
        \begin{pmatrix}
            0&1\\
            -1&0
        \end{pmatrix}\\
        &=
        \frac{1-i}{2}\begin{pmatrix}
            \overline{a}+i\overline{b}&\overline{c}-i\overline{d}\\
            \overline{d}-i\overline{c}&\overline{b}+i\overline{a}
        \end{pmatrix}\\
        &=Q^\star. 
    \end{split}
\end{equation}
One can then perform a unitary transformation to retrieve the complex conjugate relation. Also, we can check that different elements in $Q$ are i.i.d. complex Gaussian variables. 
Based on the element representation of Eq.~(\ref{eq:lue_with_ups}), it is obvious that adjacent matrix elements are independent. For the remaining pairs,
\begin{equation}
    \begin{split}
        \mathbb{E}[Q_{11}Q_{22}]&=\mathbb{E}[(a+ib)(b+ia)]\\
        &=\mathbb{E}[ab-ba+ib^2+ia^2]=0,\\
        \mathbb{E}[Q_{11}\overline{Q_{22}}]&=\mathbb{E}[(a+ib)\overline{(b+ia)}]\\
        &=\mathbb{E}[(a+ib)(\overline{b}-i\overline{a})]\\
        &=\mathbb{E}[a\overline{b}+b\overline{a}+i|b|^2-i|a|^2]=0
    \end{split}
\end{equation}
holds due to $\mathbb{E}[ab]=\mathbb{E}[a\bar b]=\mathbb{E}[a^2]=\mathbb{E}[b^2]=0$ and $\mathbb{E}[|a|^2]=\mathbb{E}[|b|^2]$. 
The other pair $Q_{12},\,Q_{21}$ can be ascertained to be independent in the same way. Since diferent $2\times2$ blocks are obviously independent, the matrix ensemble is identified as the LUE.
Therefore, the action of $\Upsilon$ has essentially no effects on the LUE blocks.

Finally, the LSE parts in Eq.~(\ref{eq:without_up}) are transformed by $\Upsilon$ into an assembly of blocks in the following form: 
\begin{equation}
    W_{4\times 4}=
    \frac{1-i}{2}\begin{pmatrix}
        A+iD&B-iC\\
        C-iB&D+iA
    \end{pmatrix},
\end{equation}
where $A,B,C$ and $D$ are $2\times 2$ spin representations of four independent quaternions arising from subsystem $l\cup r$. Because quaternions themselves are invariant under half-integer spin time reversal action, this matrix is invariant under the integer spin time reversal action of $\sigma_y\otimes\sigma_y K$.
After a unitary transformation by $\Upsilon$, we can make $W_{4\times 4}$ real:   
\begin{equation}
    \begin{split}
        K \Upsilon W_{4\times 4} \Upsilon^\dag K&=\Upsilon^3 KW_{4\times 4}K\Upsilon\\
        &=\Upsilon \left(\sigma_y\otimes\sigma_y K\right) W_{4\times 4} \left(K\sigma_y\otimes\sigma_y \right) \Upsilon^\dag\\
        &=\Upsilon W_{4\times 4} \Upsilon^\dag,
    \end{split}
\end{equation}
where we have used $\Upsilon^3=\overline{\Upsilon}=\Upsilon^{-1}$ and $\Upsilon^2=-\sigma_y\otimes\sigma_y$. To prove the independence between the matrix elements, we note that the covariance matrices before and after the transformation are related by
\begin{equation}
    \begin{split}
        &\mathbb{E}\left[\left[\Upsilon W_{4\times 4}\Upsilon^\dag\right]_{jj'}\left[\Upsilon W_{4\times 4}\Upsilon^\dag\right]_{ll'}\right]\\
        &=\mathbb{E}\left[\left[\Upsilon W_{4\times 4}\Upsilon^\dag\right]_{jj'}\overline{\left[\Upsilon W_{4\times 4}\Upsilon^\dag\right]_{ll'}}\right]\\
        &=\sum_{\alpha,\beta,\gamma,\delta}\Upsilon_{j,\alpha}\Upsilon^\dag_{\beta,j'}\overline{\Upsilon_{l,\gamma}\Upsilon^\dag_{\delta,l'}}\mathbb{E}\left[\left[W_{4\times 4}\right]_{\alpha,\beta}\overline{\left[W_{4\times 4}\right]_{\gamma,\delta}}\right].
    \end{split}
\end{equation}
A brute-force calculation shows that originally no pairs of 
different matrix elements are correlated, i.e., $\mathbb{E}\left[\left[W_{4\times 4}\right]_{\alpha,\beta}\overline{\left[W_{4\times 4}\right]_{\gamma,\delta}}\right]=\delta_{\alpha,\gamma}\delta_{\beta,\delta}$. 
The unitarity of $\Upsilon$ reveals that the new covariance matrix is also an identity:
\begin{equation}
    \begin{split}
        &\sum_{\alpha,\beta,\gamma,\delta}\Upsilon_{j,\alpha}\Upsilon^\dag_{\beta,j'}\overline{\Upsilon_{l,\gamma}\Upsilon^\dag_{\delta,l'}}\delta_{\alpha,\gamma}\delta_{\beta,\delta}\\
        &=\sum_{\alpha,\beta} 
        \Upsilon_{j,\alpha}\Upsilon^\dag_{\beta,j'}\overline{\Upsilon_{l,\alpha}\Upsilon^\dag_{\beta,l'}}\\
        &=\sum_{\alpha,\beta} 
        \Upsilon_{j,\alpha}\Upsilon^\dag_{\alpha,l}\Upsilon_{l',\beta}\Upsilon^\dag_{\beta,j'}\\
        &=\delta_{j,j'}\delta_{l,l'}.
    \end{split}
\end{equation}
Therefore, all the elements in $\Upsilon W_{4\times 4}\Upsilon^\dag$ are i.i.d. real Gaussian variables, and thus the LSE blocks are turned into LOE blocks. 

In summary, in the presence of $\Upsilon$, the matrix ensemble is identified as
\begin{gather}
    \left[\bigoplus_{\alpha:R_{-1}}\frac{\mathbb{1}_{d_\alpha}}{d_\alpha} \otimes \mathrm{LSE}_\alpha^{d_Ld_\alpha\times d_Rd_\alpha}\right]\nonumber\\
    \oplus \nonumber\\
    \left[\bigoplus_{\alpha:R_0}\frac{\mathbb{1}_{d_\alpha}}{d_\alpha} \otimes \left(\mathrm{LUE}_\alpha^{d_Ld_\alpha\times d_Rd_\alpha}\oplus\sigma_y\overline{\mathrm{LUE}_\alpha^{d_Ld_\alpha\times d_Rd_\alpha}}\sigma_y\right)\right]\nonumber\\
    \oplus\nonumber\\
    \left[\bigoplus_{\alpha:R_{1}}\frac{\mathbb{1}_{d_\alpha}}{d_\alpha} \otimes \Upsilon^\dag\mathrm{LOE}_\alpha^{d_Ld_\alpha\times d_Rd_\alpha}\Upsilon\right].
\end{gather}
The joint distribution of the entanglement spectrum is the same as Eq.~(\ref{eq:distrubution_antiunitary}).
The proof of Theorem \ref{thm:espectrum} has now been completed.

\section{A theorem on ensemble types and their numbers}

In our Letter, the main effect of symmetry fractionalization by $\Omega$ and $\Upsilon$ with nontrivial cohomology classes is changing the ensembles via changing the dimensions of irreps as well as their indicators. The latter effect is significant when $G=G_0\rtimes \mathbb{Z}_2^\mathcal{T}$ is considered. However, symmetry fractionalization has a limitation on changing the ensemble types and their numbers in the threefold-way decomposition:
\begin{theorem}
    \label{thm:ensemble_nogotheorem}
    Without $\Upsilon$, the numbers of increased LOE blocks (with possible degeneracies counted) 
    by the action of $\Omega$, is not larger than the numbers of increased LSE blocks. 
    Within $\Upsilon$, the numbers of increased LSE blocks by the action of $\Omega$ with a nontrivial cohomology, is not larger than the numbers of increased LOE blocks. 
\end{theorem}

\begin{proof}

We first discuss the notations. 
Assume the symmetry with TRS: $G=G_0\rtimes \mathbb{Z}_2^\mathcal{T}$.
First, $\iota_\mathrm{rep}$ and $\iota'_\mathrm{rep}$ denote the indicators of certain representations with trivial and nontrivial cohomologies of $G_0$.
Without $\Omega$, we get the sets of $\iota_\alpha=\pm\omega(t,t),0$ linear irreps of $G_0$ to be $R_\pm,\,R_0$, respectively. After applying $\Omega$, likewise, we obtain the sets of $\iota'_\alpha=\pm\omega(t,t),0$ projective irreps of $G_0$ to be $R'_\pm,\,R'_0$, respectively.

In general, the indicator for the regular representation of group $G_0$ is
\begin{equation}
\begin{split}
    \iota^{(')}_\mathrm{reg}&=\frac{1}{|G_0|}\sum_{g\in G_0}\omega(\tilde{g},g)\chi(\tilde{g}g)\\
    &=\sum_{\tilde{g}g=e}\omega(\tilde{g},g).
\end{split}
\end{equation}
If we assume $\Omega$ with a possibly nontrivial cohomology is applied, since $\omega(\tilde{g},g)\in U(1)$ and the indicator $\iota_\mathrm{reg}$ is always real, the indicator $\iota'_\mathrm{reg}$ cannot be larger than $\iota_\mathrm{reg}$.

From a different perspective, the indicator for the regular representation can be decomposed into the indicators of irreps. Here, the sum can be decomposed into the irrep types:
\begin{equation}
\begin{split}
    \iota^{(')}_\mathrm{reg}&=\sum_{\alpha} d_\alpha \iota_\alpha\\
    &=\sum_{\alpha\in R_+} d_\alpha-\sum_{\alpha\in R_-} d_\alpha\\
    &=D_+-D_-
\end{split}
\end{equation}
where $D_\pm=\sum_{\alpha\in R_\pm}d_\alpha$, which means the number of LOE, and/or LSE blocks with counting degeneracies as duplications. Similarly we can define $D'_\pm=\sum_{\alpha\in R'_\pm}d_\alpha$. 
If we assume $\omega(t,t)=1$, $\iota^{(')}_\mathrm{reg}$ is the difference of the number of LOE blocks with duplications minus the number of LSE blocks with duplications. The discussion above implies $\iota_\mathrm{reg}\geq \iota'_\mathrm{reg}$. This completes the proof of the 
theorem. 
\end{proof}

\section{Other representations}

In our Letter, we used the regular representation to construct $G$-symmetric bases, but we will show that similar constructions are possible for other representations and that the threefold-way decomposition holds.

\subsection{$G=G_0$}

Let us start from the case of general linear representations in the absence of TRS. We make a minimal assumption that the symmetry representation on $a\cup b$ takes the form of $D_a(g)\otimes D_b (g)$, which obviously covers the case of arbitrary on-site representations. We can decompose each representation into irreps:
\begin{equation}
D_a=\bigoplus_\alpha \mathbb{1}_{a_\alpha} \otimes D^\alpha,\;\;\;\;
D_b=\bigoplus_\alpha \mathbb{1}_{b_\alpha} \otimes D^\alpha,
\label{DaDb}
\end{equation}
where $a_\alpha$ and $b_\alpha$ are non-negative integers. The projector onto the global symmetric subspace spanned by $|\Psi\rangle=D_a(g)\otimes D_b (g)|\Psi\rangle$ $\forall g\in G$ reads
\begin{equation}
P=\frac{1}{|G|}\sum_{g\in G}D_a(g)\otimes D_b (g).
\label{P}
\end{equation}
Substituting Eq.~(\ref{DaDb}) into Eq.~(\ref{P}), 
we obtain
\begin{equation}
\begin{split}
    P&=\bigoplus_{\alpha,\beta}\mathbb{1}_{a_\alpha}\otimes \left[\sum_{g\in G}\frac{1}{|G|} D^\alpha(g)\otimes D^\beta(g) \right]\otimes \mathbb{1}_{b_\beta}.\\
\end{split}
\end{equation}
Applying the grand orthogonality theorem and dropping the zero subspaces, we end up with
\begin{equation}
P=\bigoplus_\alpha \mathbb{1}_{a_\alpha}\otimes\mathbb{1}_{b_{\bar \alpha}},
\end{equation}
where $\bar \alpha$ denotes the irrep related to $\alpha$ via complex conjugation. Therefore, a random state within the subspace of $P$ is determined by i.i.d. coefficients $c_{L_\alpha,\alpha,R_\alpha}$ with  the degree (i.e., number of labels) of $L_\alpha$ ($R_\alpha$) being $a_\alpha$ ($b_{\bar\alpha}$). Moreover, the explicit eigenstate of $\sum_{g\in G}D^\alpha(g)\otimes D^{\bar\alpha}(g)$ should be a Bell (maximally entangled) state of qudits with $d=d_\alpha$ (again due to the grand orthogonality theorem), we can thus write down a symmetric random state as
\begin{equation}
|\Psi\rangle = \sum_{\alpha,L_\alpha,R_\alpha} c_{L_\alpha,\alpha,R_\alpha} \frac{1}{\sqrt{d_\alpha}}\sum^{d_\alpha}_{j_\alpha=1}|L_\alpha j_\alpha j_\alpha R_\alpha \rangle.
\label{cLaR}
\end{equation}
The above procedure of identifying the global pure state could be considered as a highly abstract version (and in the dual picture of irreps) of symmetry concentration. According to Eq.~(\ref{cLaR}), the reduced density matrix on $a$ 
\begin{equation}
\rho_a=\bigoplus_\alpha \frac{\mathbb{1}_{d_\alpha}}{d_\alpha} \otimes {\rm LUE}_\alpha^{a_\alpha\times b_{\bar\alpha}}. 
\label{ra}
\end{equation}

Replacing $D_{a,b}$ by $\mathcal{D}_{a,b}$, which are two projective representations with opposite cohomology classes, we can perform similar calculations and end up with formally the same result (\ref{ra}). This straightforward generalization relies on the fact that if $\mathcal{D}^\alpha$ is a projective irrep, $\overline{\mathcal{D}^\alpha}$ will be a projective irrep with respect to the opposite cohomology class.

\subsection{$G=G_0\rtimes \mathbb{Z}_2^\mathcal{T}$}

We move on to incorporate the TRS. First let us return to usual group representations. Considering $\overline{D^\alpha(\tilde g)}$ ($\tilde g = tgt$) as an irrep of $G_0$ in terms of $g$, we obtain a possibly different irrep labeled by $\alpha^\star$, whose relation to $D^\alpha$ is determined by the twisted index. Accordingly, if $D_{a,b}$ is compatible with TRS, i.e., $D(t)D_{a,b}(g) D(t)=D_{a,b}(\tilde g)$, the numbers of irreps $\alpha$ and $\alpha^\star$ (if $\alpha\neq\alpha^\star$) should be equal. More concretely, we can write down $D_a(t)$ (similar for $D_b(t)$) as 
\begin{equation}
\begin{split}
    D_a(t) &= \left[\bigoplus_{\alpha: R_+} \mathbb{1}_{a_\alpha}\otimes\mathbb{1}_{d_\alpha}\right]\oplus
\left[\bigoplus_{\alpha: R_0}  \sigma_x\otimes\mathbb{1}_{a_\alpha}\otimes\mathbb{1}_{d_\alpha}\right]\\
&\quad \oplus \left[ \bigoplus_{\alpha: R_-} Y_{a_\alpha}\otimes Y_{d_\alpha} \right] K,
\end{split}
\end{equation}
where $K$ denotes the complex conjugation, $R_{\pm,0}$ refer to the sets of irreps with $\iota_\alpha=\pm ,0$ (where $R_0$ only includes one component in $\{\alpha,\alpha^\star\}$), $\sigma_x$ exchanges $\alpha$ and $\alpha^\star$ subspaces, and $Y_n =\sigma_y\otimes\mathbb{1}_{n/2}$. Therefore, imposing $D_a(t)\otimes D_b(t)|\Psi\rangle = |\Psi\rangle$ causes further constraints on $c_{L_\alpha,\alpha,R_\alpha}$ in Eq.~(\ref{cLaR}), leading to
\begin{equation}
\begin{split}
    \rho_a=&\left[\bigoplus_{\alpha: R_+} \frac{\mathbb{1}_{d_\alpha}}{d_\alpha} \otimes {\rm LOE}_\alpha^{a_\alpha\times b_{\bar\alpha}}\right]\\ 
    &\oplus\left[\bigoplus_{\alpha:R_0} \frac{\mathbb{1}_{d_\alpha}}{d_\alpha} \otimes \left({\rm LUE}_\alpha^{a_\alpha\times b_{\bar\alpha}}\oplus \overline{  {\rm LUE}_\alpha^{a_\alpha\times b_{\bar\alpha}} }\right) \right]\\
    &\oplus\left[\bigoplus_{\alpha: R_-} \frac{\mathbb{1}_{d_\alpha}}{d_\alpha} \otimes {\rm LSE}_\alpha^{a_\alpha\times b_{\bar\alpha}}\right]. 
\end{split}
\label{rat}
\end{equation}
Similar derivations hold true if the unitary parts of the representations are projective.

Finally we discuss the case of TRS fractionalization. The only difference turns out to be
\begin{equation}
\begin{split}
    D_a(t) &= \left[\bigoplus_{\alpha: R_+} Y_{a_\alpha}\otimes\mathbb{1}_{d_\alpha}\right]\oplus
\left[\bigoplus_{\alpha: R_0}  \sigma_x\otimes Y_{a_\alpha}\otimes\mathbb{1}_{d_\alpha}\right]\\
&\quad \oplus \left[ \bigoplus_{\alpha: R_-} \mathbb{1}_{a_\alpha}\otimes Y_{d_\alpha} \right] K,
\end{split}
\end{equation}
leading to formally the same result (\ref{rat}), as long as we reinterpret $R_\pm$ as the set of irreps with $\iota_\alpha=\pm\omega(t,t)$. These results clearly reduce to Theorem~1 for regular representations.


\end{document}